\documentclass[conference]{IEEEtran}
\IEEEoverridecommandlockouts

\usepackage{graphicx}
\usepackage{balance}
\usepackage[linesnumbered,ruled,vlined]{algorithm2e}
\usepackage{algpseudocode}
\usepackage{color}
\usepackage{subfigure}
\usepackage{cite}
\usepackage{graphicx}
\usepackage{textcomp}
\usepackage{xcolor}
\def\BibTeX{{\rm B\kern-.05em{\sc i\kern-.025em b}\kern-.08emT\kern-.1667em\lower.7ex\hbox{E}\kern-.125emX}}
\usepackage{amssymb}
\usepackage{amsfonts}
\usepackage{amsmath}
\usepackage{ntheorem}
\usepackage[colorlinks,citecolor=blue,urlcolor=blue,bookmarks=true,hypertexnames=true]{hyperref}
\usepackage{cleveref}

\newtheorem{problem}{Problem}

\newtheorem{definition}{Definition}

\newtheorem{lemma}{Lemma}
\newtheorem*{proof}{Proof}
\usepackage[ruled]{algorithm2e}
\usepackage{lipsum}

\usepackage{etoolbox}

\begin{document}

\title{On Efficient and Scalable Time-Continuous Spatial Crowdsourcing --- Full Version}

\author{
\IEEEauthorblockN{Ting Wang\IEEEauthorrefmark{1}\IEEEauthorrefmark{2},
Xike Xie\IEEEauthorrefmark{1}\IEEEauthorrefmark{2},
Xin Cao\IEEEauthorrefmark{3},
Torben Bach Pedersen\IEEEauthorrefmark{4},
Yang Wang\IEEEauthorrefmark{1}\IEEEauthorrefmark{2} and
Mingjun Xiao\IEEEauthorrefmark{1}\IEEEauthorrefmark{2}
\IEEEauthorblockA{\IEEEauthorrefmark{1}School of Computer Science and Technology, University of Science and Technology of China, China}
\IEEEauthorblockA{\IEEEauthorrefmark{2}Suzhou Institute for Advanced Study, University of Science and Technology of China, China }
\IEEEauthorblockA{\IEEEauthorrefmark{3} University of New South Wales, Australia}
\IEEEauthorblockA{\IEEEauthorrefmark{4}  Aalborg University, Denmark }
Email: \{tingwt@mail, xkxie@, angyan@, xiaomj@\}ustc.edu.cn, xin.cao@unsw.edu.au, tbp@cs.aau.dk}
}

\maketitle
\begin{abstract}
The proliferation of advanced mobile terminals opened up a new crowdsourcing avenue, spatial crowdsourcing, to utilize the crowd potential to perform real-world tasks.
In this work, we study a new type of spatial crowdsourcing, called time-continuous spatial crowdsourcing (TCSC {\it in short}).
It supports broad applications for long-term continuous spatial data acquisition, ranging from environmental monitoring to traffic surveillance in citizen science and crowdsourcing projects.
However, due to limited budgets and limited availability of workers in practice, the data collected is often incomplete, incurring data deficiency problem.
To tackle that, in this work, we first propose an entropy-based quality metric, which captures the joint effects of incompletion in data acquisition and the imprecision in data interpolation. Based on that, we investigate quality-aware task assignment methods for both single- and multi-task scenarios. We show the NP-hardness of the single-task case, and design polynomial-time algorithms with guaranteed approximation ratios. We study novel indexing and pruning techniques for further enhancing the performance in practice. Then, we extend the solution to multi-task scenarios and devise a parallel framework for speeding up the process of optimization. We conduct extensive experiments on both real and synthetic datasets to show the effectiveness of our proposals.
\end{abstract}


\section{Introduction}

Spatial crowdsourcing or crowdsensing refers to harnessing human knowledge or sensors of participants' smart phones to retrieve qualitative or quantitative details related to physical locations of crowdsourced tasks.
For conventional spatial crowdsourcing, task assignment and fulfillment are often ``atomic'' in that either they are fully executed or not at all
 \cite{kazemi2012geocrowd,deng2013maximizing,cheng2015reliable,tong17,to2015server,DBLP:conf/icde/ChengLCS17,deng2015task,DBLP:journals/tetc/ZhangYLT19,DBLP:conf/icde/ChengLCS17,DBLP:conf/icde/ChenCZC19,cheng2016task,DBLP:conf/icde/ChengCY19}.
In this work, we consider a special type of spatial crowdsourcing, called {\it time-continuous spatial crowdsourcing} (TCSC {\it in short}).
TCSC is different, due to its temporal continuity, such that a spatial crowdsourced task takes long to finish, necessitating the time-sharing collaboration of multiple workers.
It finds broad applications in capturing the presence and duration of environmental features, e.g., air/water pollution monitoring\cite{DBLP:conf/iwqos/LiWZ19} and traffic surveillance\cite{DBLP:journals/sensors/ZhangZ0DG19}-\cite{DBLP:journals/wpc/SabirMS19}, which are prevalent in citizen science projects~\cite{citizenscience}.

For example, in Fig.~\ref{fig:intro}, a crowdsourcer would like to analyze the microbial content in the water for a period.  Upon receiving the task, the TCSC server looks up the records of preregistered workers' spatiotemporal information. Indicating the undertaking of workers giving to the TCSC server, the registered spatiotemporal information consists of workers' available time slots, working regions, and so on, e.g., $\{ worker_1, \langle place~A, 1-2pm\rangle, \langle place~B, 7-8pm\rangle, \dots \}$. The task is then decomposed into a set of subtasks. Each subtask corresponds to a specific time slot and location. Subtasks are assigned to appropriate workers according to the assignment policy. Workers finish the assignment, e.g., probing environmental values, and send their results to the server. The crowdsourced results are aggregated and delivered to the crowdsourcer.


Quality is essential for such applications. It is infeasible to accomplish a crowdsourced task for all time slots, due to limited budgets and availability of workers. So, the probed data in crowdsourced results is inherently incomplete.
Interpolation (or extrapolation) alleviates the data incompletion problem by inferring unproved values with the probed ones.
However, the interpolation error may further affect the data precision, incurring the so-called data deficiency problem.
Ignorance of the facts of data incompletion and data imprecision would cause unreliable crowdsourced results.
Thus, it is of paramount importance to consider the data quality problem in the TCSC setting. To this end, we propose a general entropy-based metric for summarizing the amount of incompleteness and impreciseness of the crowdsourced results, which enables quality-aware TCSC assignment and balances the plannable expense and observable essence.

\begin{figure}
\centering
\includegraphics[height=1.8in,width = 1\columnwidth]{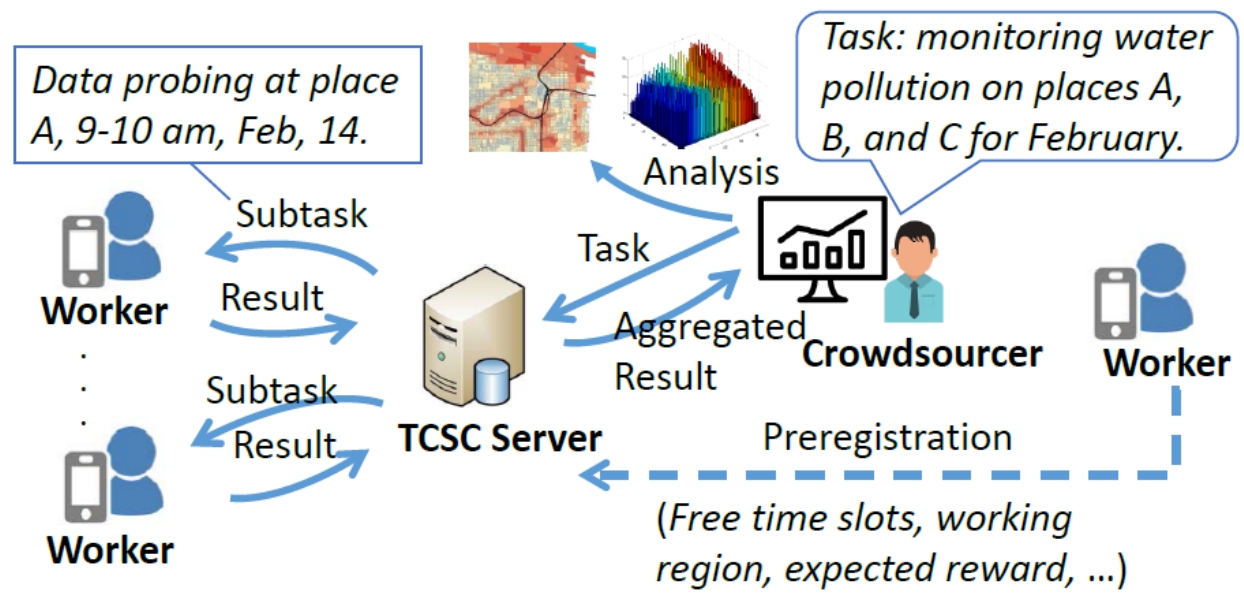}
\caption{General TCSC Framework}
\label{fig:intro}
\end{figure}
There is a substantial difference between the TCSC problem and existing task assignment problems in spatial crowdsourcing. For example, existing works mostly focused on maximizing the total number of completed tasks~\cite{kazemi2012geocrowd}, on maximizing the number of performed tasks for an individual worker~\cite{deng2013maximizing}, or on maximizing the reliability-and-diversity score of assignments~\cite{cheng2015reliable}.
These solutions cannot be directly applied, as none of them look into the temporal continuous nature and corresponding quality issues of the TCSC problem.
To our best knowledge, we are the first to study the TCSC problem. 


Nevertheless, the computational overhead of quality-aware TCSC assignment is high. Even a simplified version of the problem, i.e., single TCSC task assignment, is NP-hard, as shown in Section~\ref{sec:single}.
In this work, we study how the quality-aware TCSC assignment can be handled in an efficient and scalable way.
For ease of presentation, we start with the simplified version, single-task assignment, with the target of maximizing the task quality under budget constraints.
We prove that the problem is NP-hard and further show that it can be approximated by a polynomial-time solution with guaranteed ratios. We also devise novel pruning and indexing techniques for efficiency enhancement.
Based on that, we introduce the multi-task case, where technical challenges arise in handling the correlations between a given set of single tasks. We devise a parallel framework by distributing multiple correlated tasks to independently running computation cores, so as to maximally utilize the independence between tasks.

The main contributions of this paper are as follows.

\begin{itemize}
\item We propose and formalize the novel TCSC problem.
\item We prove that the problem is NP-hard and therefore study approximation algorithms for accelerating the task assignment with quality guarantees, for both single- and multi-task scenarios.
\item We investigate novel indexing and pruning techniques for the efficiency of the single-task case.
\item We devise an efficient parallelization framework for the multi-task case, by breaking ties of correlated task groups with devised synchronization mechanisms.
\item We conduct extensive experiments on synthetic and real data to evaluate the efficiency and scalability.
\end{itemize}

The rest of this paper is organized as follows.
Section \ref{sec:preliminaries} introduces preliminaries, including concepts, quality metrics, and properties.
Section~\ref{sec:single} starts with the single-task assignment.
Section \ref{sec:multiple} extends the solution to the multi-task scenario.
Section \ref{sec:experiments} evaluates our proposals with extensive experiments.
Section \ref{sec:related} presents related works.
Section \ref{sec:conclusion} concludes the paper.
Table~\ref{tab:notations} summarizes the symbols and notations used.


\newcommand{\tabincell}[2]{
}

\begin{table}[h!]
\centering
\caption{Summary of Notations}
\label{tab:notations}
\begin{tabular}{|c|l|}
\hline
Symbol & Meaning \\ \hline\hline
$\mathcal{T} = \{ \tau_i \}$ & a set of tasks \\
\hline
$\tau_i = \{ \tau_i^{(j)}\}_{j = 1}^m$ & a set of $m$ subtasks of task $\tau_i$\\
\hline
$\tau_i^{(j)}$ & a subtask of $\tau_i$ at $j$-th time slot\\
\hline
$c(\tau)$ and $c(\tau^{(j)})$ & cost of task $\tau$ and subtask $\tau^{(j)}$ \\
\hline
$W = \{ w_i \}_{i = 1}^{n}$ & a set of $n$ workers \\
\hline
$w_i^{(j)}$ & worker $w_i$ at $j$-th time slot \\
\hline
$m$ & number of subtasks of a task \\
\hline
$n$ & number of workers \\
\hline
$|a,b|_i$ & distance in time between slots  $a$ and $b$\\
\hline
$q(.)$  & quality metric function \\
\hline
$p(.)$ & subtask finishing probability function \\
\hline
$\rho_{err}(.)$  & interpolation error ratio function\\
\hline
\end{tabular}
\end{table}

\section{Preliminaries}
\label{sec:preliminaries}


In this section, we introduce basic concepts, propose the quality metric, and investigate its properties. For ease of presentation, we start with the single-task case (Sections~\ref{sec:preliminaries} and \ref{sec:single}), and extend it to the multi-task case (Section~\ref{sec:multiple}).
\subsection{Basic Concepts}
\label{subsec:basic}

\newtheorem{def1}{Definition}
{\bf TCSC tasks and subtasks.}
A single TCSC task $\tau$ has its location $\tau.loc$ and duration $\tau.dur$.
According to the batch size that tasks arrive in, the duration consists of at most $m$ equal-sized time slots. Thus, $\tau$ can be represented by a set of subtasks $\{\tau^{(j)}\}_{j = 1}^m$ such that each subtask $\tau^{(j)}$ takes $\tau.loc$ as its location, and the corresponding time slot as its duration. Formally, $\tau^{(j)}.loc = \tau.loc$ and $\tau^{(j)}.dur = \frac{\tau.dur}{m}$.

{\bf Worker.} Let $W = \left\{ w_1, w_2, ..., w_n\right\}$ be a set of $n$ workers. Each worker is registered with a set of consecutive states to the SC server, indicating whether s/he is online for providing crowdsourcing services. A worker $w_i$'s temporal state can be represented by $w_i^{(j)}$, indicating the availability of worker $w_i$ at time slot $t_j$.

{\bf Task Assignment.}
Task assignment is the mapping of workers to subtasks.
In Fig. \ref{fig:assign}, there are $5$ subtasks in $\tau$. At each time slot, there exist a set of workers.
In this example, workers $w_2^{(2)}$ and $w_4^{(4)}$ are assigned to subtasks $\tau^{(2)}$ and $\tau^{(4)}$, respectively.
$\tau^{(3)}$ is not assigned to any worker as none are available at that time slot. $\tau^{(1)}$ is not mapped to any worker because of cost and budget limits, e.g., workers at time slot $1$ are far away from $\tau^{(1)}$.

\begin{figure}[h]
\centering
\includegraphics[width = 0.42\columnwidth]{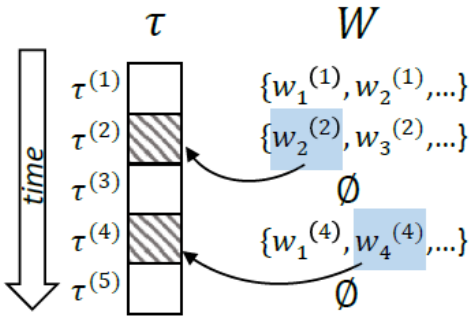}
\caption{An Example of TCSC Task Assignment ($m=5$)}
\label{fig:assign}
\end{figure}

{\bf Cost.} The cost of a subtask $\tau^{(j)}$ is denoted as $c(\tau^{(j)})$. Following the common setting of spatial crowdsourcing, we assume the travel cost of a subtask $c(\tau^{(j)})$ is the Euclidean distance from the location of a subtask $\tau^{(j)}$ and the assigned worker $w$\footnote{If considering traveling distances as costs\cite{tong17}~\cite{todsurvey}, the nearest worker is usually selected in order to minimize the cost of taking a subtask. It is also possible to take the second nearest neighbor or the latter in multi-task scenarios, as covered by Section~\ref{sec:multiple}.}. For simplicity, we assume the unit cost of all workers is the same. Our work is general w.r.t. the type of cost.
The cost of a task is the summation of all its subtasks' costs, i.e., $c(\tau) =\sum^m_{j=1}c(\tau^{(j)})$.


\subsection{Quality Metric}
\label{subsec:quality}
{Due to the limited budget or the availability of workers, it is necessary to measure the {\it quality} of a TCSC task, given the fact that a task cannot be fully assigned.
There can be two possible states for a subtask, {\it executed} and {\it unexecuted}, corresponding to whether the subtask value is {\it probed} by assigned workers or {\it interpolated} by other probed subtasks. For example, in Fig.~\ref{fig:assign}, $\tau^{(2)}$ and $\tau^{(4)}$ are executed subtasks, while the other $3$ are unexecuted and need to be interpolated.
}
Initially, all subtasks are ``null''. If some subtasks are assigned and probed, their states are changed from ``unexecuted'' to ``executed''. Meanwhile, remaining unexecuted subtasks can be inferred from the executed subtasks by interpolation (or extrapolation).
So, the quality metric of a TCSC task should, 1) distinguish the quality between a probed item and an interpolated item; 2) be universal for covering the integrative action of different stated subtasks for the overall quality measurement.

{\bf Quality Metric.} For condition 1), we define the concept of {\it subtask finishing probability}, based on the calculation of potential interpolation errors. For condition 2), we utilize an entropy function that returns a real-valued score for conveniently indicating the amount of inaccuracies in accordance with specific assignment strategies, as shown below.
\begin{definition}
({\bf Task Quality.}) Let $\tau$ be a task consisting of $m$ subtasks, $\tau = \{\tau^{(j)}\}_{1 \leq j \leq m}$. Each subtask $\tau^{(j)}$ is associated with a finishing probability $p^{(j)}$. The quality of $\tau$, denoted by $q(\tau)$, is:
\begin{equation}
    \label{eqn:quality}
    q(\tau) = -\sum_{j=1}^{m}p^{(j)}\log_2\big{(} p^{(j)}\big{)}
    \end{equation}
\end{definition}


Next, we introduce the finishing probability $p^{(j)}$ for subtask $\tau^{(j)}$, that serves as building blocks for the quality metric.

{\bf Subtask Finishing Probability.} The finishing probability of a TCSC task equals $1$, if all its subtasks are done, in correspondence to an ideal case that all subtasks are executed.
In practice, for a task consisting of $m$ equally sized  time slots, the finishing probability $p^{(j)}$ for each subtask $\tau^{(j)}$ is at most $\frac{1}{m}$. Without losing generality, we can thus use an error ratio $\rho_{err}$ to measure the amount of information loss caused by interpolation errors.
\begin{align}
\label{eqn:p}
p^{(j)} = \frac{1}{m} (1 - \rho_{err}(\tau^{(j)}))
\end{align}
Accordingly, the probability $p^{(j)}$ equals $0$ for the ``null'' case, representing the zero knowledge about the subtask, and equals $\frac{1}{m}$ for the executed case, representing the total information gain of finishing the subtask.
Next, we show the calculation of error ratio $\rho_{err}$ for an interpolated subtask.


{\bf Interpolation Error Ratio.}
A common way of inferring a missing value from a discrete set of known values is known as inverse distance interpolation \cite{gao2006voting}\cite{gao2006adaptive}\cite{babak2009statistical}, which averages the values of its $k$ nearest neighbors, i.e., $k$ nearest subtasks on the timeline of a TCSC task.

Intuitively, the error ratio is proportional to the distances between the interpolated value and its neighboring values~\cite{gao2006voting}\cite{gao2006adaptive}.
{
Between two subtasks $\tau^{(i)}$ and $\tau^{(j)}$,
the temporal distance is denoted as $|\tau^{(i)}, \tau^{(j)}|_i$, referring to the absolute difference of $\tau^{(i)}$ and $\tau^{(j)}$'s timestamps.
For example, in Fig.~\ref{fig:assign}, we have $|\tau^{(1)},\tau^{(2)}|=1$ and $|\tau^{(2)},\tau^{(4)}|=2$.
Then, the interpolation error of an unexecuted subtask can be evaluated by the distances from a set of executed subtasks.
Assume function $S_{kNN}(.)$ returns the set of $k$ executed subtasks with the smallest distances.
An unexecuted subtask $\tau^{(j)}$ can thus be interpolated by $S_{kNN}(\tau^{(j)})$, with the interpolation error measured by the error ratio function $\rho_{err}$ as follows.
}


\begin{align}
\label{eqn:rho}
\footnotesize
\rho_{err}(\tau^{(j)}) = \frac{\sum_{e \in S_{kNN}(\tau^{(j)})} |\tau^{(j)},e|_i}{k\cdot m}
\end{align}



The value range of the error ratio $\rho_{err}$ is from $0$ to $1$. The error ratio equals $100\%$, if none of the subtasks of a task is executed. Accordingly, a lower error ratio value is achieved, if the target subtask has more proximate executed subtasks (i.e., smaller interpolation distances). {For example, in Fig.~\ref{fig:assign}, $\tau^{(1)}$'s $2$-NN results are $\tau^{(2)}$ and $\tau^{(4)}$, whose distances from $\tau^{(1)}$ are $1$ and $3$, respectively.} So, $\rho_{err}(\tau^{(1)})$ can be calculated as $\frac{1+3}{2*100} = 0.02$. Since $\tau^{(2)}$ is an executed subtask, its error ratio equals zero~\footnote{
It is possible that a subtask $\tau^{(j)}$ has less than $k$ nearest neighbors, e.g., at the starting stage of subtask assignment.
If $|S_{kNN}| < k$, we let $|\tau^{(j)},e|_i$ be $m$, indicating the largest possible interpolation distance.}.
Similarly, we can derive the error ratios and the subtask finishing probabilities for all the subtasks. Based on that, we can calculate the task quality, by Equation~\ref{eqn:quality}.

In summary, the concept of entropy is adopted to quantify the quality of crowdsourced results. The value of $q(\tau)$ ranges from $0$, i.e., the lowest information degrees (none of the subtasks are executed), to $log_2m$, i.e., the highest information degree (all subtasks are executed).

{\bf Extension for Worker Reliability.}
Our quality metric, i.e., Equation~\ref{eqn:quality}, is general in addressing the reliability issues~\cite{cheng2015reliable}, where workers are not assumed to be entirely trustable. Instead, each worker $w_i$ is assumed to have a reliability score, represented by $\lambda_i \in [0, 1]$.
Incorporating the reliability, the subtask finishing probability can be defined by both worker confidence and potential interpolation error. In particular, given subtask $e$, the finishing probability of an assigned task is $\frac{\lambda_e}{m}$, where $\lambda_e$ is the reliability of the worker assigned to subtask $e$.

Due to the incorporation of worker reliability, for an interpolated subtask $\tau^{(j)}$, the maximum finishing probability is no longer $\frac{1}{m}$, but the product of $\frac{1}{m}$ and the average of $\{\lambda_e\}_{e \in S_{kNN}(\tau^{(j)})}$. The extended form of subtask probability is as Equation~\ref{eqn:p2}.
\begin{align}
\label{eqn:p2}
p^{(j)} = \frac{1}{m} \cdot [\frac{\sum_{e \in S_{kNN}(\tau^{(j)})} \lambda_e}{k}  -
\rho_{err}(\tau^{(j)})]
\end{align}
Accordingly, the error ratio is determined by the summation of interpolation distances weighted by the worker reliability $\lambda_e$ and then divided by $k\cdot m$, similar to Equation~\ref{eqn:rho2}. The extended form of error ratio is as Equation~\ref{eqn:rho2}.
\begin{align}
\label{eqn:rho2}
\footnotesize
\rho_{err}(\tau^{(j)}) = \frac{\sum_{e \in S_{kNN}(\tau^{(j)})} \lambda_e \cdot |\tau^{(j)},e|_i}{k\cdot m}
\end{align}
Notice that probability $p^{(j)}$ equals $0$, if $S_{kNN}(\tau^{(j)})$ is an empty set, meaning that null of subtasks are currently done.
The probability $p^{(j)}$ is $\frac{\lambda_{\tau^{(j)}}}{m}$, if $\tau^{(j)}$ is an executed subtask.
If $\tau^{(j)}$ is interpolated, and the reliability of each worker for executing $S_{kNN}(\tau^{(j)})$ equals $1$, Equation~\ref{eqn:rho2} degenerates into Equation~\ref{eqn:rho}.

{\bf Summary.} We have shown the effectiveness of the metric in capturing the quality of various scenarios. In the sequel, we show two important properties of the metric, submodularity and non-decreasingness, enabling efficient task assignment algorithms with approximation guarantees.
\subsection{Properties}
\label{subsec:property}

We derive the properties of the task quality metric. We first recall the definition of submodular  functions.
\begin{definition}
\label{submodular}
({\bf Submodular Function \cite{cover2012elements}.}) Let $S$ be a finite set, and $2^S$ be the power set of $S$. A submodular function is a set function $f: 2^S \rightarrow \mathbb{R} $, if for any $X,Y\subset S$ satisfying $f(X\cap Y)+f(X\cup Y) \leq f(X)+f(Y)$.
\end{definition}

\begin{lemma}
\label{lem:double}
({\bf Composite Submodular Functions \cite{Lin}}.)
Assume a set $V$ and a function $h: 2^V \rightarrow \mathbb{R}$ that returns a real value for any subset $S \subseteq V$.
If $h$ is a non-decreasing submodular function and $\phi$ is a non-decreasing concave function $\phi: \mathbb{R} \rightarrow \mathbb{R}$, then function $\phi(h(S))$ is non-decreasing and submodular.
\end{lemma}
\begin{lemma}
\label{lem:q}
Function $q(\tau)$ is non-decreasing and submodular.
\end{lemma}
\begin{proof}
The quality metric function $q(.)$ is a composite function of the entropy function and the finishing probability function, i.e., $p(.)$.
It can be proved that function p(.) is submodular and non-decreasing, as proved by Lemmas 7 and 8 in Appendix.
More, the entropy function is known as non-decreasing and concave. According to Lemma~\ref{lem:double}, the quality metric function $q(.)$ is non-decreasing and submodular w.r.t. $\tau$.
\end{proof}
The submodularity property captures the effect of marginal benefits decreasing in the measuring the TCSC quality. As the number of workers on finishing a task is increasing, the marginal value of adding a new worker is decreasing.
Also, the submodular property of the quality metric enables efficient optimization, e.g., submodular function maximization, where constant factor approximation algorithms are often available~\cite{article}. In the sequel, we discuss how the property is utilized and generalized for optimization problems in the TCSC task assignment applications.

%
%
%
%
%
%
%
%

\section{SINGLE TASK ASSIGNMENT}
\label{sec:single}
We proceed to study the single TCSC task assignment problem.
In particular, we consider maximizing the quality of a single task with budget limits.
We prove that the problem is NP-hard.
For efficient task assignment, we devise heuristic algorithms that approximate the optimization targets in polynomial time with quality guarantees.
Nevertheless, we show that the polynomial solution incurs overheads unaffordable for the real-time task assignment scenarios.
Therefore, we propose novel indexing and pruning techniques to further enhance the task assignment efficiency.

We formalize and analyze the problem in Section~\ref{subsec:pdef}. We devise the approximation algorithm in Section~\ref{sec:single_alg}. We study efficient indexing and pruning techniques in Section~\ref{subsec:eff}.
\subsection{Problem Definition and Analysis}
\label{subsec:pdef}

\begin{problem}
{\bf Single-task Quality Maximization with Fixed Budgets (sQM {\it in short}).}
Given a TCSC task $\tau$ and a fixed budget $b$, the sQM problem is to find an assignment for $\tau$, such that the quality $q(\tau)$ is maximized, and the cost $c(\tau)$ does not violate the given budget $b$.
\label{6}
\begin{align}
    Maximize~~ & q(\tau)  \nonumber \\
    subject~~to~~ &\sum_{j=1}^{m} c(\tau^{(j)}) \leq b \nonumber
 \end{align}


\end{problem}
To solve the sQM problem is equivalent to evaluate a task assignment matrix yielding the maximum quality.
To give a sense of the size of the solution space, we assume a given set of $n$ workers and a task consisting of $m$ subtasks/timeslots. Thus, there could be an exponential number of possible worker-and-subtask assignment pairs $O(m^n)$.
Then, we prove that the $sQM$ problem is NP-hard, by Lemma~\ref{lem:np}.

\begin{lemma}
\label{lem:np}
    The $sQM$ problem is NP-hard.
\end{lemma}
\begin{proof}
It is well known that maximizing a submodular function under a cardinality constraint (i.e., selecting at most $k$ elements) is NP-hard~\cite{Nemhauser1978}. If we consider the workers have unit cost, the budget constraint in our problem becomes the cardinality constraint. It has been proved in Lemma~\ref{lem:q} that the quality function $q(\tau)$ (i.e., the objective function to be maximized) is submodular, and thus a special case of our $sQM$ problem is also NP-hard.
\end{proof}
\subsection{Approximation Algorithms}
\label{sec:single_alg}

\newcommand{\Tau}{\mathrm{\it T}}

Hereby, we provide a suboptimal solution with guaranteed approximation ratios, based on the submodular property of the quality metric, as discussed in Section~\ref{subsec:property}.
Since the problem is of budgeted maximization for a submodular function, we can have a heuristic algorithm, which repeatedly selects an element (e.g. a subtask) that maximizes the quality increment until the budget is exceeded. The process is detailed in Algorithm~\ref{alg:single}.

\begin{algorithm}
\label{alg:single}
\caption{Single task assignment algorithm} 
\small
\LinesNumbered
\KwData{$b>0$, a set of workers $W$, a task $\tau$}
Initialize the states of subtasks $\{\tau^{(j)} \in \tau \}$ as {\tt NULL}  and initialize $T_{cur}'$ and $T_{cur}$ as two empty sets\;
For each subtask $\tau^{(j)}$ get the corresponding cost     $c(\tau^{(j)})$\;
Execute the subtask $\tau^{(h)}$ yielding the highest quality but not exceeding the budget, $T_{cur}' \leftarrow \{ \tau^{(h)} \}$\;
\While{$C(\tau) \leq b$}
{
    \For{ $ \tau^{(j)} \in \tau - T_{cur}$ }
    {
            Compute $\frac{q(\Tau_{cur} \cup  \tau^{(j)}) - q(\Tau_{cur} )}{c(\tau^{(j)})}$\;
    }
    $\tau^{(*)} \leftarrow argmax\left\{ \frac{\Delta q(\tau)}{c(\tau^{(*)}) } : \tau^{(*)} \in \tau \right\}$\;
    Update $\tau^{(*)}$'s state to {\tt Executed}\;
    $T_{cur} \leftarrow T_{cur} \cup \tau^{(*)}$\;
}
return $T_{cur}'$ or $T_{cur}$ with the highest quality as the final result\;
\end{algorithm}
Hence, the heuristic value is defined as the quality increment divided by the corresponding cost.
Let the currently assigned set of subtasks be $\Tau_{cur}$.
At each iteration, the algorithm greedily selects a subtask $\tau^{(*)}$ from the set $\tau - \Tau_{cur}$, such that the heuristic value is maximized.
Formally, the greedy rule is to find a subtask $\tau^{(*)}$ as follows.
\begin{align}
\small
\tau^{(*)} & = \underset{\tau^{(j)}}{argmax } \frac{\Delta q(\tau)}{c(\tau^{(j)}) } = \underset{\tau^{(j)}}{argmax} \frac{q(\Tau_{cur} \cup  \tau^{(j)}) - q(\Tau_{cur})}{c(\tau^{(j)})} \nonumber
\end{align}

\begin{figure*}
  \center
  \includegraphics[height=1.9in,width = 2\columnwidth]{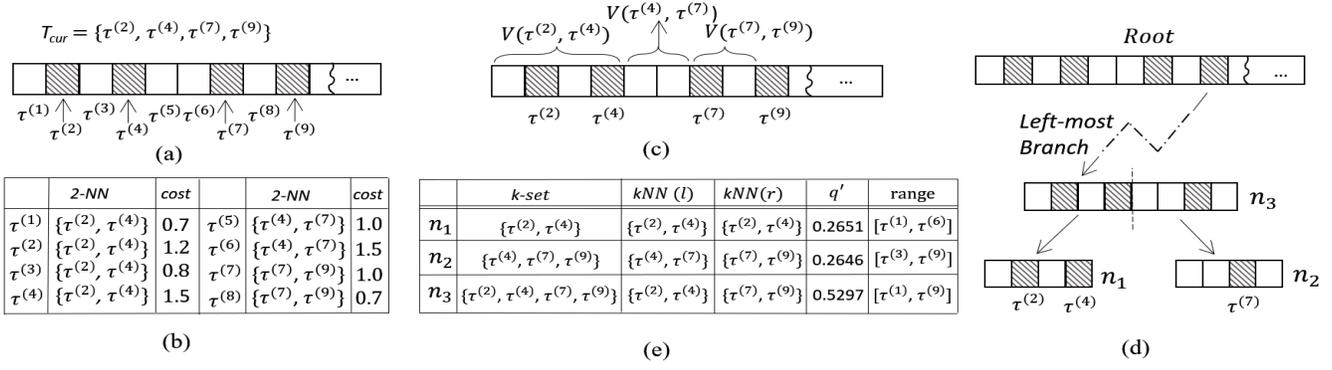}

\vspace{-5pt}
  \caption{An Example. ($k=2$, $t_s = 4$, $m=100$): (a) Current state $T_{cur}$; (b) $2$-Nearest Neighbors; (c) $1$-dimensional Order-$2$ Voronoi Diagram; (d) Tree-structured Index; (e) Auxiliary Information.}

  \label{fig:voronoi}
\end{figure*}

We use an example to illustrate the process.
Assume that there are $4$ executed subtasks, i.e., $T_{cur} = \{ \tau^{(2)}, \tau^{(4)}, \tau^{(7)}, \tau^{(9)} \}$, of a TCSC task, represented by shaded slots, as shown in Fig.~\ref{fig:voronoi} (a).
Algorithm~\ref{alg:single} enumerates all remaining slots/subtasks.
At each iteration, a slot is selected for tentative execution in order to find the one with the highest heuristic value.
If $\tau^{(1)}$ is chosen and tentatively executed, $\rho(\tau^{(1)})$ is reduced to $0$ and the heuristic value increment is $0.0016$.
However, if $\tau^{(5)}$ is tentatively executed, $\tau^{(6)}$'s $2$-NN result is changed from $\{\tau^{(4)}, \tau^{(7)}\}$ to $\{\tau^{(5)}, \tau^{(7)}\}$.
Therefore, both $\rho(\tau^{(5)})$ and $\rho(\tau^{(6)})$ should be recalculated for getting the quality and the heuristic values.
The process is repeated for all the unexecuted slots. Finally, $\tau^{(1)}$ is selected, since it derives the maximum heuristic value in this example.
By setting the greedy strategy as such, Algorithm~\ref{alg:single} guarantees a $(1-1/\sqrt{e})$ approximation to the optimal solution, as shown in \cite{article}.

{\bfseries Complexity Analysis}. First, the number of iterations, i.e., the outer loop of Algorithm~\ref{alg:single} (line $4$), is at the level of $O(m)$. Second, the inner loop (line $5$) is at the level of $O(m)$, since one needs to try all $m$ subtasks and get their heuristic values in order to find the one maximizing the overall quality.
Third, for each trial of a subtask, one needs to calculate the corresponding heuristic value, i.e., the overall quality increase of implementing the subtask, according to the quality metric function. The overall quality increase is the summation of individual quality increments of all other $m-1$ subtasks, so the complexity is $O(m)$ (line 6).
Fourth, for quality increment of a individual subtask, $k$ timeslots in the neighborhood should be visited.
In our implementation, we maintain a sorted list for subtasks that are sorted in the ascending order of the corresponding time slots.
During the query evaluation, $O(log(m))$ cost is used for finding the nearest assigned subtask, and then $O(k)$ cost is used for refining the exact $k$-NN.
Therefore, the total time complexity is $O(m^3log(m))$.

In summary, the approximation algorithm gives a polynomial alternative for tackling the NP-hard optimization problem. However, the computational overhead makes it impractical in real-time task assignment scenarios. In the sequel, we propose a series of techniques for better efficiency and scalability.
\subsection{Efficient Heuristic Value Calculation}
\label{subsec:eff}


%


The idea of accelerating the algorithm is in two parts.
First, for finding the maximum heuristic value, how to avoid unnecessary enumeration of all $m$ subtasks.
Second, the calculation of a heuristic value refers to the summation of partial qualities of all subtasks. For the heuristic value calculation, how to maximally reuse the computation so as to avoid unnecessary checking of all time slots.


{\bf Locality of $k$-NN Searching.} We try to scale down the problem by considering the locality of $k$-NN searching. We observed from Fig.~\ref{fig:voronoi} (b) that if two slots are proximate, their $k$-NN results tend to be similar. For example, $\tau^{(1)}$ and $\tau^{(2)}$ share the same $k$-NN results, i.e., $\{\tau^{(2)}, \tau^{(4)}\}$.
Theoretically, the solution space of $k$-NN searching over the $m$ subtasks is a one-dimensional order-$k$ Voronoi diagram. The domain space is a one-dimensional interval, i.e., from $1$ to $m$. The diagram splits $m$ slots into disjoint intervals, called Voronoi cells, such that the $k$-NN searching provides the same result if queries are within the same Voronoi cell. In Fig.~\ref{fig:voronoi} (c), for example, the Voronoi cell $V(\tau^{(2)}, \tau^{(4)})$ covers slots from $\tau^{(1)}$ to $\tau^{(4)}$, meaning that all slots from $\tau^{(1)}$ to $\tau^{(4)}$ take $\{\tau^{(2)}, \tau^{(4)}\}$ as their $2$-NN results.

Such a structure facilitates the algorithm evaluation in two aspects. First, the time cost for $k$-NN searching can be reduced from $O(log(m))$ to constant time $O(1)$, as the diagram pre-computes the solution space for the $k$-NN queries. Second, the diagram accelerates the calculation of heuristic values.
For the tentative execution of a subtask, the heuristic value is calculated by enumerating all other slots which takes $O(m)$ (Algorithm~\ref{alg:single}, line 6).

Recall that the finishing probability of an unexecuted subtask depends on its $k$-NN interpolation. So, the value of finishing probability of an unexecuted subtask does not change, if the order-$k$ Voronoi cell to which it belongs does change, with the tentative execution.
This way, the problem of quality increment calculation is transformed to the reformulation of Voronoi cells w.r.t. the tentative subtask execution, which can be handled locally, since the diagram handles such updates locally~\cite{voronoi-book}. We cover more details below.

The technical challenge arising is that there could be a large number of order-$k$ Voronoi cells, making the gains of local computation not worthy of the overhead of the Voronoi diagram construction.
In particular, the average number of order-$k$ cells is $O(k(m-k))$~\cite{voronoi-book}.

{\bf Approximated One-dimensional Voronoi Diagram.} To handle that, we propose an approximate version of one-dimensional order-$k$ Voronoi diagram. The idea is to use an aggregated binary tree for Voronoi cell indexing and Voronoi diagram approximation (Fig.~\ref{fig:voronoi} (d)).
In the tree, each node represents a time segment $[l, r]$, where $l$ and $r$ are for the two slots on the segment's left and right ends, respectively. The root node is the interval of the entire $m$ slots. For each node, we store the auxiliary information in the form of a quadruples, i.e., $\langle k$-$set, knn(l), knn(r), q' \rangle$, as shown in Fig.~\ref{fig:voronoi} (e).

The $k$-set of a node is the union of $k$-NN results for all its offsprings.
$knn(l)$ and $knn(r)$ are $k$-NN results of the two end slots. The $k$-NN results are sorted in ascending order of the distance to $l$ (or $r$), so that the distance from $l$ (or $r$) to its $k$-th nearest neighbor can be fast retrieved, denoted by $k_{max}(l)$ (or $k_{max}(r)$). So, we can derive the influence range of a node as $[max(1, l-k_{max}(l)), min(m, r+k_{max}(r))]$, such that the quadruple of the node can be affected if a tentatively executed slot is within the influence range.
$q'$ is the partial quality value of the node.
For subtask $\tau^{(j)}$, its partial quality equals $p^{(j)}log(p^{(j)})$.
The quality of a node is the summation of partial qualities of all subtasks in its offspring.
This way, the quality value can be fast retrieved by querying the upper level nodes of the tree structure. Upon updating, only necessary subtrees are retrieved and revised, following the style of updating an aggregated tree-structure.

{\bf Maximum Heuristic Value Calculation.}
With the tree structure, it takes much less computational overhead than enumerating all $m$ slots. The process is implemented by traversing the tree in a best-first manner with an associated heap. The elements of the heap are tree nodes which are sorted in descending order of the upper bounds. The higher an element ranks in the heap, the more likely it corresponds to the maximum heuristic value increment.

Next, we study how to upper bound the effect of a tentative insertion.
A tentative insertion affects the $k$-NN interpolation results so the quality value varies. The effect can be categorized into {\it inter-node} and {\it intra-node} cases\footnote{We consider the leaf node for ease of presentation. The calculation of upper bounds for non-leaf nodes can be done in a bottom-up manner, as the construction process of an aggregated tree-structure.}.

For the intra-node case, if an unexecuted slot is tentatively executed, of the same node, another unexecuted slot's $k$-th NN distance is reduced to $1$ at most. It corresponds to the extreme case that the slot is next to the tentative executed slot, whereas the distance between them is $1$. So, the interpolation error ratio can be lower bounded as follows.
\begin{align}
\small
\label{eqn:lb}
\rho_{err}(\tau^{(j)}) \geq \frac{1}{k\cdot m} [\sum_{e \in S_{(k-1)NN}(\tau^{(j)})} |e, \tau^{(j)}|_i + 1]
\end{align}


Therefore, the upper bound of a node's quality can be derived by the lower bound of $\rho_{err}$ (Equation~\ref{eqn:lb}), since the value of $q(.)$ is inversely proportional to that of $\rho_{err}(.)$. So, the upper bound of a node's heuristic value is calculated by the maximum quality change divided by the minimum cost of all unexecuted subtasks in the node.

For the inter-node case, a tentatively executed slot of a node would also change the $k$-NN interpolation result of other nodes. For example, in Fig.~\ref{fig:voronoi} (d), we have two leaf nodes, $n_1$ and $n_2$. If $\tau^{(3)}\in n_1$ is tentatively executed, $\tau^{(5)}\in n_2$'s second nearest neighbor changes from $\tau^{(7)}$ to $\tau^{(3)}$. The corresponding quality value change should also be incorporated and updated.

At each iteration, the top element of the heap is popped up and the child nodes are inserted. If the top element is a time slot/subtask, the exact heuristic value increase can be obtained. Suppose the value is $\theta$.
Then, all other elements in the heap whose upper bounds are below $\theta$ can be pruned. The process repeats until the heap is empty. The slot that gives the maximum heuristic value is then obtained.

The advantages of using the tree structure are two-fold. First, using best-first searching for obtaining the slot with the potentially maximum heuristic value is expected to take $O(log(m))$ time.
Second, the locality of $k$-NN searching is reflected by the decomposition of tree leaf nodes.
Suppose that calculating the updated heuristic value of a leaf node takes constant time.
Then, instead of enumerating $O(m)$ slots for calculating the updated heuristic value, one takes $O(log(m))$ time to retrieve and update the relevant leaf nodes, reducing the corresponding computation overhead from $O(m)$ to $O(log(m))$.
The overall cost is $O(mlog^3(m))$.


{\bf Tree Construction.}
We consider the tree construction in an incremental manner.
Suppose that a time slot $e$ is to be tentatively executed, meaning that $e$ is with the maximum current heuristic value.
The process of tree construction is triggered accordingly, and is revoked recursively.
At each iteration, we test if a subtree will be affected. If yes, we update the associated quadruples and forward the updates to the descendants. Otherwise, the entire subtree is skipped.

We summarize two cases to determine whether the current node will be affected and therefore updated by the tentative execution. Based on that, we can disqualify irrelevant nodes at a higher level to save the overhead of the tree construction.


\begin{itemize}
\item \emph{Case 1.} A node will be affected, if the tentatively executed slot is within the influence range of the node.
\item \emph{Case 2.} A node will not be affected, if its parent node is not affected, by $e$.
\end{itemize}


By doing so, we can find appropriate nodes for updating.
The updates can thus be propagated to the leaf level so that the node splitting should be handled.

{\bf Splitting and Stopping Conditions.}
During splitting, a node is decomposed into two sub-nodes, with its time segment split into two equal sized sub-segments.
The quadruple of the node can thus be partially inherited by its sub-nodes.
For example, the left sub-node inherits the $knn(l)$ of its parent, because they share the same left end slot.
In particular, we consider two stopping conditions.
\begin{itemize}
\item \emph{Condition 1.} For a node with segment $[l, r]$, splitting stops, if $knn(l) = knn(r)$.
\item \emph{Condition 2.} If the length of a node's segment is smaller than a pre-specified threshold, $t_s$, splitting stops.
\end{itemize}
Condition $1$ guarantees that current segment belongs to the same order-$k$ Voronoi cell, so that there is no need for further splitting. Condition $2$ limits the depth of the tree structure to $\lceil log_{t_s}(m) \rceil$. It serves as a knob for tuning the approximation accuracies so as to control the construction overhead of the tree structure.
The correctness of Conditions $1$ is guaranteed by Lemma~\ref{lem:cond1} in Appendix.

\section{MULTIPLE TASK ASSIGNMENT}
\label{sec:multiple}

The multi-task assignment problem is essential to the practical deployment of crowdsourcing platforms, where multiple tasks are submitted, scheduled, and executed, simultaneously.
However, the computational overhead of multi-task assignment is high, even its simplified version, i.e., single-task assignment, is NP-hard.
Even with the approximation solution, the algorithm has to iteratively retrieve a subtask from all given $|\mathcal{T}|$ tasks that maximizes the heuristic value, making the algorithm scale quadratically with $|\mathcal{T}|$, which is not scalable for handling a large number of tasks.
A practical way of handling multi-task case is to fully exploit the hardware capabilities of the TCSC server, with parallel computing techniques.

In this section, we study the multi-task assignment scenario, by considering two variants regarding the settings of optimization targets.
The first variant evaluates the overall quality by the summation of qualities of individual tasks, that belong to the given task set.
The second variant improves the overall quality by reinforcing the ``weakest'' single task, i.e., maximizing the minimum single task quality.
Both variants are on improving the overall quality of the given set of tasks with budget constraints\footnote{A dual version of our problem can be minimizing the task costs with quality constraints. It can be handled with the primal-dual method~\cite{DBLP:conf/approx/ChekuriK04}, which reduces the problem to the one studied in this work.}. 
In particular, we use $q_{sum}$ and $q_{min}$ to represent the optimization target functions, respectively. 
\begin{figure*}[ht]
  \center
\includegraphics[height=1.1in,width = 2\columnwidth]{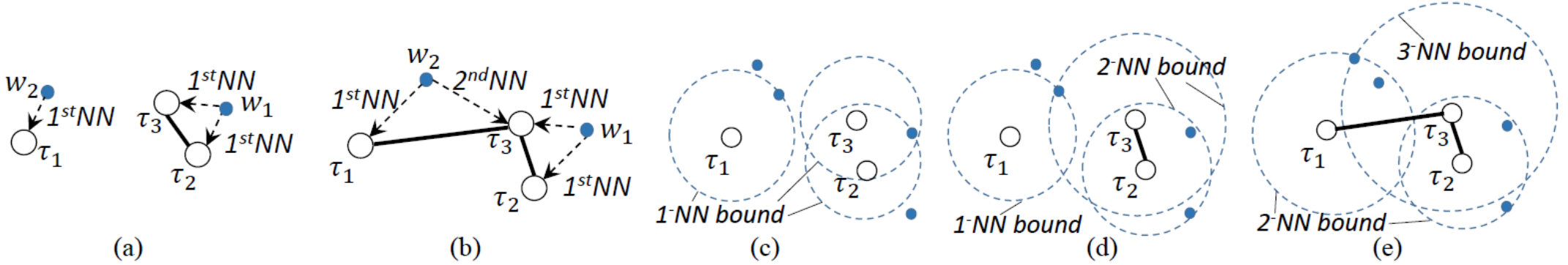}
  \caption{Worker Conflicting and Group-level Parallelization: an example of considering travel distances between workers and subtasks as costs.}
  \label{fig:group}
\end{figure*}

\subsection{Maximizing Summation Quality}
\label{subsec:maxsum}

The first optimization target is on maximizing the summation quality of all tasks as following.
\begin{definition}
{\bf Summation Quality.} Given a set of tasks $\mathcal{T} = \{\tau_1, \tau_2, ...\}$, we define the summation quality as:
    \begin{equation}
    \label{eqn:Quality_M}
           q_{sum}(\mathcal{T})= \sum_{i=1}^{|\mathcal{T}|}q(\tau_i \big{|} \tau_i \in \mathcal{T})
    \end{equation}
\end{definition}
\begin{problem}
{\bf Multiple-task Summation Quality Maximization with Fixed Budgets (MSQM {\it in short}).} Given a set of tasks $\mathcal{T} = \{\tau_1, \tau_2, ...\}$, the MSQM problem is to find an assignment for the tasks in $\mathcal{T}$, such that the summation quality is maximized, and the overall cost $\sum c(\tau_i)$ does not exceed the given budget $b$.
\vspace{-10pt}
   \begin{equation}
   \vspace{-5pt}
     \begin{aligned}
       Maximize~~q_{sum}(\mathcal{T})\\
       subject~to~\sum_{i=1}^{|\mathcal{T}|} c(\tau_i) \leq b
     \end{aligned}
     \vspace{-3pt}
   \end{equation}
 \end{problem}
We can prove that the MSQM problem is NP-hard, by reducing it from the sQM problem, whose NP-hardness is proved by Lemma~\ref{lem:np}.




\begin{lemma}
$q_{sum}(.)$ is submodular and non-decreasing.
\end{lemma}
\begin{proof}
We have proved that $q(.)$ is non-decreasing and submodular, by Lemma~\ref{lem:q}. The summation function is known as both convex and concave.
The $q_{sum}(.)$ function is a composite function of a summation function and $q(.)$ function. So, the lemma is proved, according to Lemma~\ref{lem:double}.
%
%
\end{proof}

Based on the properties of submodularity and non-decreasingness, the framework of single task assignment, i.e., Algorithm~\ref{alg:single}, can be applied for handling the multi-task assignment case. The heuristic value is set as the increase of the summation quality divided by the corresponding cost (of a tentatively selected subtask), following the same greedy strategy. Then, the algorithm is to iteratively retrieve a subtask from all given $|\mathcal{T}|$ tasks that maximizes the heuristic value, so that the solution space is $|\mathcal{T}|$ times the size of the single task case, making the algorithm scale quadratically with $|\mathcal{T}|$.
The time complexity is $O(|\mathcal{T}|^2mlog^3(m))$.
To improve the scalability, we aim to derive a parallelization framework for distributing the calculation workload onto multiple computation cores. Ideally, each task can be running independently on different cores, so that the time cost would be $\frac{|\mathcal{T}|}{\text{\# of cores}}$ times that of running a single task, assuming the value of $|\mathcal{T}|$ is larger than the number of cores. But, there occurs correlations between tasks, if two subtasks running on different cores ``compete'' for one worker at some time slot. It happens because it is possible that two subtasks choose the same worker with lowest costs for minimizing the budget decrement.
We call this {\it worker conflicting}, as exemplified in Fig.~\ref{fig:group} (a), where there are three tasks ($\tau_1$ to $\tau_3$) and two workers ($w_1$ and $w_2$). There exist conflicts between $\tau_2$ and $\tau_3$, since they both take $w_1$ as the worker with the lowest cost.

\subsubsection{Group-level Parallelization}
We can have a graph of independent groups, if taking each task as a node and drawing an edge between any two conflicted tasks. If a group of nodes do not have any connections with other groups, it is an independent group. The optimization process of such a group is independent of others and therefore independent groups can be run in parallel.
For example, in Fig.~\ref{fig:group} (a), tasks can be divided into two independent groups, $\{\tau_1\}$ and $\{\tau_2,\tau_3\}$, if any pair of tasks do not compete for workers with lowest costs.

Is it sufficient for considering the conflicts on the lowest costs for deriving the independent groups?
The answer is NO. For example, in Fig.~\ref{fig:group} (b), $\tau_2$ and $\tau_3$ have conflicts on $w_1$, so $\tau_3$ opts for $w_2$, who ranks as the worker with the second lowest cost. Unfortunately, $\tau_3$ further have conflicts with $\tau_1$, which has $w_2$ as the worker with the lowest cost. In general, if two subtasks belonging to different tasks are to be executed, one of them has to choose a worker with the second lowest cost, or even the latter.

One may get the independence graph by gradually expanding the searching regions.
For example, assume costs are calculated by the travelling distances from workers to tasks~\cite{todsurvey}\cite{tong17}. Then, a subtask takes the nearest worker as the worker with the lowest cost. For a given task, we call the circle centered at the task's position with the distance between the task and its nearest neighbor as its $1$-NN bound.
The independence graph can be obtained by the following steps.
Initially, we draw $1$-NN bounds for each task, as shown in Fig.~\ref{fig:group} (c). An edge between $\tau_2$ and $\tau_3$ is added, since the two tasks share the same worker that causes conflicts. Next, we draw $2$-NN bounds for both $\tau_2$ and $\tau_3$, so that there exist enough workers for being assigned to the two conflicting tasks, as shown in Fig.~\ref{fig:group} (d).
However, it shows that another worker is within $\tau_1$'s $1$-NN bound and $\tau_3$'s $2$-NN bound, so the two tasks have conflicts and the edge between them is added.
After that, we draw $2$-NN bound for $\tau_1$ and $3$-NN bound for $\tau_3$, as shown in Fig.~\ref{fig:group} (e). In general, if a node of the independence graph is with degree $d$, the $\text{(d+1)}$-NN bound should be drawn. The process repeats until no conflicts are detected.
A drawback of the gradual expanding method is on incurring large groups and heavyweight computation tasks, deteriorating the parallelization performance.

\subsubsection{Task-level Parallelization}
Hereby, we devise the task-level parallelization framework, as depicted in Fig.~\ref{fig:parallel}. We set a thread pool, with a master thread and a set of worker threads waiting for tasks to be concurrently executed.
The master thread is for maintaining the thread pooling on heartbeat monitoring, conflicting controlling, scheduling, and logging. To support the functionalities of the master thread, there are several associated data structures, {\it Heartbeat Table}, {\it Conflicting Table}, and {\it Logging Table}.
\begin{figure}[ht]
\centering
\includegraphics[width = 1\columnwidth]{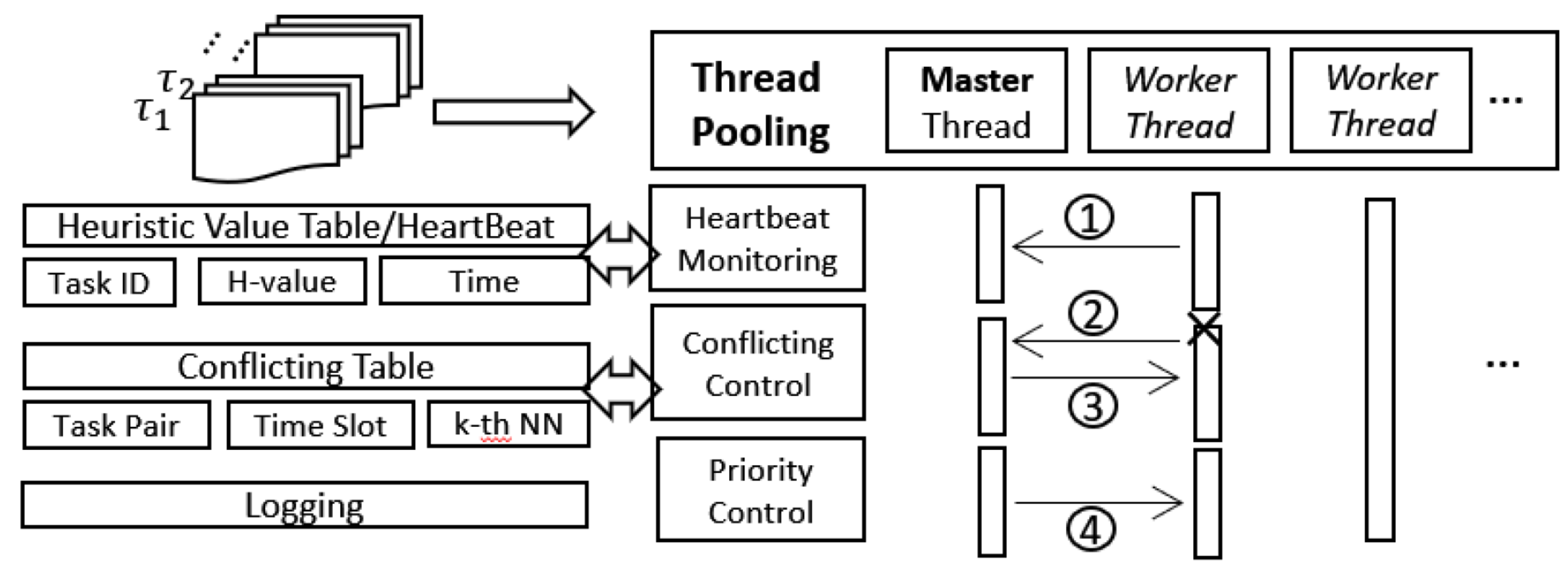}
  \caption{Task-level Parallelization: \textcircled{1} Periodically Sending Heartbeats to Master Thread; \textcircled{2} Reporting to Master Thread that Conflicts Detected; \textcircled{3} Looking Up Conflicting Table and Heartbeat Table, then Ask Worker Thread to Continue or Suspend; \textcircled{4} Adjusting Priorities of Worker Threads.}

  \label{fig:parallel}
\end{figure}

Heartbeat Table stores periodically reported heuristic values from currently executed tasks. Logging Table traces the historical records of Heartbeat table.
Conflicting Table stores a series of records for breaking the ties of conflicts. Assume three tasks, $\tau_1$, $\tau_2$, and $\tau_3$, are conflicted at time slot $t$. The information, including the conflicting task sets and corresponding conflicting time slot, is stored at the Conflicting Table of the master thread. In this example, a tuple $\langle \{\tau_1,\tau_2,\tau_3\}, t, 1 \rangle$ is recorded in the conflicting table.
Here, $1$ means the three tasks are to compete for the worker of $1$-NN.
Then, during the task processing, if $\tau_1$ is to execute conflicting slot $t$, its associated thread sends a message to inform the master thread. Upon receiving the message, the master thread looks up the Conflicting table and the Heartbeat table to retrieve the current heuristic values of $\tau_1$, $\tau_2$, and $\tau_3$.
If $\tau_1$'s current heuristic value is higher than $\tau_2$ and $\tau_3$, it continues with the execution of slot $t$. Meanwhile, the master thread checks the availability of workers and updates the record on the conflicting table, by changing the field ``{\it k-th NN}'' from $1$ to $2$, so that $\tau_2$ and $\tau_3$ would compete for the worker with the $2^{nd}$ lowest cost next time, because the first one has been taken by $\tau_1$.
Otherwise $\tau_1$ is suspended and the process continues.

{\bf Discussion.}
The task-level parallel approach is deterministic, meaning that the parallelized task assignment plan is consistent with the non-parallel plan.
The master thread periodically stores and descendingly sorts the heuristic values that are collected in Heartbeat Table, so that the derived plan is the same as the serialized task execution (Algorithm~\ref{alg:single}).
This way, the parallel algorithm follows the approximation framework with guaranteed ratio.
On the other hand, in the parallel environment, it is hard to strictly control the stopping condition, i.e., the timeline when the given budget is exhausted. It is unavoidable that threads with lower heuristic values are executed earlier than those with higher values.
But this can mostly be alleviated with our priority settings.
We set priorities in accordance with the heuristic values of worker threads dynamically, so that the tasks with higher heuristic values are more likely to be processed.
This is also consistent with the greedy strategy of Algorithm~\ref{alg:single}. The priorities of worker threads are initialized as infinity to avoid thread starvation.

\subsection{Maximizing Minimum Quality}
\label{subsec:maxmin}
The second optimization target is on maximizing the minimum quality of all tasks, so that the overall quality is optimized. The problem is formalized.
\begin{definition}
\label{def:Quality_{min}}
(Minimum Quality) Given a set of tasks $\mathcal{T} = \{\tau_1, \tau_2, \dots\}$, we define the minimum quality as:
    \begin{equation}
      q_{min}(\mathcal{T})= min \left\{q(\tau_i) \big{|} \tau_i \in \mathcal{T} \right\} 
    \end{equation}
\end{definition}
\begin{problem}
    {\bf Multi-task Minimum Quality Maximization with Fixed Budgets (MMQM).} Given a set of tasks $\mathcal{T} = \{\tau_1, \tau_2, \dots\}$, the MMQM problem is to find a task assignment for each task $\tau_i \in \mathcal{T}$, such that the minimum quality is maximized, and the overall cost $\sum{c(\tau_i)}$ does not exceed the given budget $b$.
\vspace{-10pt}
       \begin{equation}
       \vspace{-10pt}
         \begin{aligned}
           Maximize~  &q_{min}(\mathcal{T})\\
           Subject~to~&\sum_{i=1}^{|\mathcal{T}|} c(\tau_i) \leq b  \\ \nonumber
         \end{aligned}
         \vspace{-10pt}
       \end{equation}
\end{problem}

We can prove the NP-hardness of the MMQM problem by reducing it to the sQM problem. The submodularity and non-decreasingness of $q_{min}(.)$ function can be proved by Lemma~\ref{lem:qmin}. Hence, the $(1-1/\sqrt{e})$ approximation ratio of $q_{min}$ is achieved by iteratively executing the selected subtask from the task yielding the minimum quality. The subtask execution follows the framework of Algorithm~\ref{alg:single}.
To fast retrieve the task with minimum quality, we maintain a heap for $\mathcal{T}$ tasks. Notice that there is no worker conflict issues for the MMQM problem, since the subtasks are executed in a sequence. So, the total time complexity is $O(mlog^3(m)log(|\mathcal{T}|))$.


\begin{lemma}
\label{lem:qmin}
    $q_{min}(\mathcal{T})$ is submodular and non-decreasing.
\end{lemma}
\begin{proof}
We have proved that $q(.)$ is non-decreasing and submodular, by Lemma~\ref{lem:q}. The minimization function is known as a concave function.
The $q_{min}(.)$ function is a composite function of a minimization function and $q(.)$ function. So, the lemma is proved, according to Lemma~\ref{lem:double}.
\end{proof}


\section{Experiments}
\label{sec:experiments}
\vspace{-5pt}
We cover the experimental setup in Section~\ref{subsec:setup}, and report the performance of our proposals in Sections~\ref{ret:quality} and \ref{ret:eff}.
\vspace{-5pt}
\subsection{Experiment Setup}
\label{subsec:setup}

{\bfseries Dataset.}\footnote{Default parameters are bolded.}
We use a real dataset\footnote{ https://www.microsoft.com/en-us/research/publication/t-drive-trajectory-data-sample/} of $10,357$ worker trajectories for representing workers' movements. For each worker trajectory, we randomly cut out a set of pieces, ranging from $1$ to $5$ time slots, as a worker's active slots.
We use a public data generator\footnote{http://chorochronos.datastories.org/sites/default/files/algorithms\\/SpatialDataGenerator.zip} to generate a series of datasets to simulate the locations of TCSC tasks, following {\bf uniform}, Gaussian, and Zipfian distributions.
For parameters of Gaussian distribution, the mean is set as the domain center and the sigma is set as the $1/6$ of the domain sidelength, so that most of generated data are within the domain space.
For Zipfian distribution, the exponent is set to $1$, which is a common setting, and the only option of the generator.
We also use a Beijing POI dataset for representing tasks' locations\footnote{https://ieee-dataport.org/documents/beijing-poi-datasets-geographical-coordinates-and-ratings}. We set the cost for an assignment to be the distance that a worker moves to the assigned task, following the common setting of spatial crowdsourcing~\cite{todsurvey}\cite{tong17}.
We vary the number of TCSC tasks to test the scalability of our proposals by setting the number of tasks as 100, {\bf 300}, and 500, respectively.
For each TCSC task, we set the task length (i.e., the number of subtasks) to 300, {\bf 500}, and  1000, respectively.
The budget is set to $\$50$, {\bf \$100}, $\$200$, corresponding to about 12.5\%, 25\%, and 50\% of the average cost of a TCSC task in the default setting. By default, $k$ is set to $3$ for the $k$-NN interpolation, $t_s$ is set to $4$. and the number of cores is set to $10$ for multi-task parallelization.

{\bfseries Implementation.} All algorithms are implemented in Java and run a PC with Intel(R) Xeon(R) CPU $E5$-$2698v4$ @ $2.20$GHz and $256$GB main memory. By default, we use $12$ cores for running experiments on multi-task assignment.
Each reported value is the average of $20$ runs.


\subsection{Results on Quality}
\label{ret:quality}

We test the effectiveness of our quality-aware task assignment method in Figure~\ref{fig:single-quality}. We compare the quality of our method, Approx, with two competitors, OPT and Rand.
OPT offers the optimal result by traversing the solution space. Rand accomplishes a task by randomly assigning a subtask to its nearest worker. The results with different data distributions is shown in Fig.~\ref{fig:single-quality} (a). In all testing, Approx achieves a high quality result which is: 1) close to the optimal result; 2) better than randomized heuristic algorithms. The randomized heuristic algorithm does not offer a deterministic solution, and therefore incurs quality fluctuations.
The gap between Approx and Rand is bigger, if the budget is smaller, as shown in Fig.~\ref{fig:single-quality} (b), which is the essential scenario in TCSC problem.

\begin{figure}[ht]
\center
\vspace{-10pt}
\subfigure[Quality vs. Distributions] {\includegraphics[height=0.8in,width=1.7in]{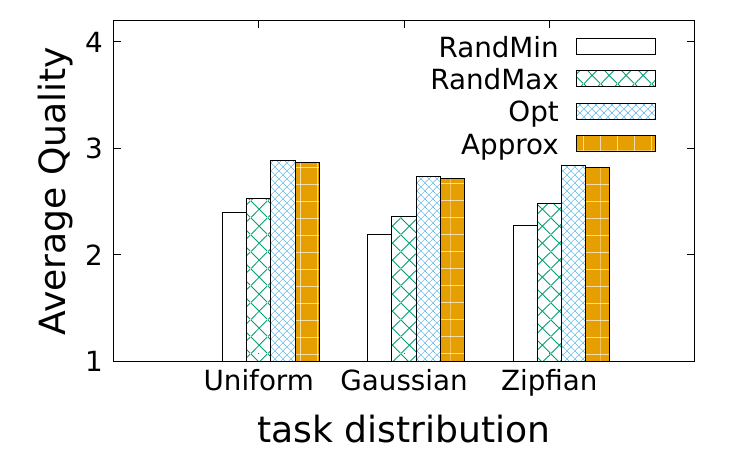}}
\subfigure[Quality vs. Budget] {\includegraphics[height=0.8in,width=1.7in]{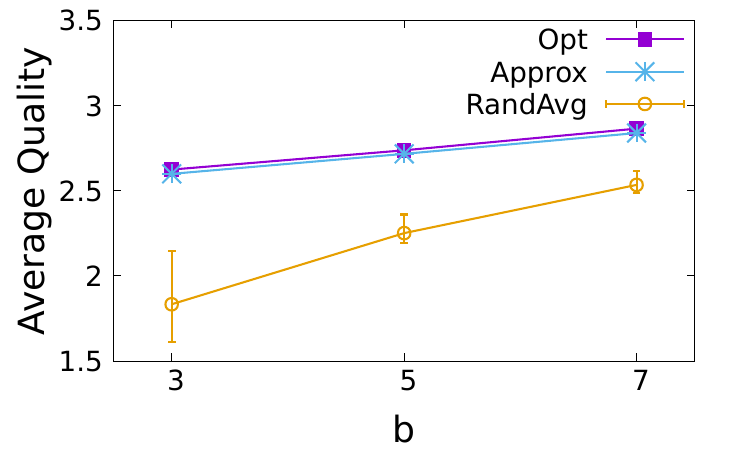}}
\vspace{-5pt}
\caption{Quality of Single-task Case}
\label{fig:single-quality}
\vspace{-5pt}
\end{figure}

We further test the results for the multi-task case in Fig.~\ref{fig:multi-quality} (a-d).
In Fig.~\ref{fig:multi-quality} (a) and (c), it can be observed that the quality of Approx is much better than its competitors, for both $q_{sum}$ and $q_{min}$ cases.
We also examine how the quality change w.r.t. the budget in Fig.~\ref{fig:multi-assignment} (b) and (d). In all cases, Approx gives much better quality than baselines. The gap between them can be smaller, if the budget is sufficiently large, which is consistent with the problem setting.

In summary, Approx offers a high quality task assignment solution, with a deterministic output and theoretical guarantees, which outperforms the baselines.

\begin{figure*}[ht]
\center
\vspace{-10pt}
\subfigure[Quality vs. Distributions ($q_{sum}$)] {\includegraphics[height=0.8in,width=1.7in]{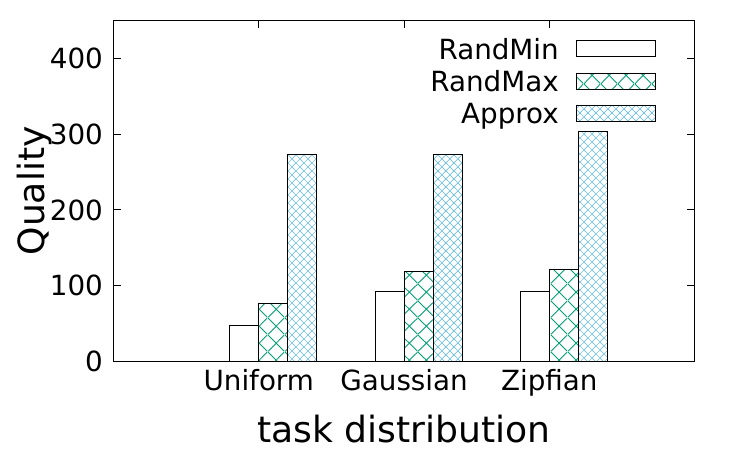}}
\subfigure[Quality vs. Budgets ($q_{sum}$)] {\includegraphics[height=0.8in,width=1.7in]{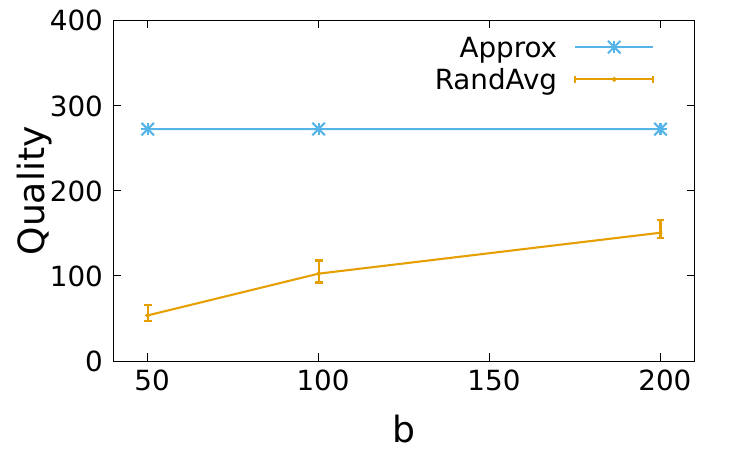}}
\subfigure[Quality vs. Distributions ($q_{min}$)] {\includegraphics[height=0.8in,width=1.7in]{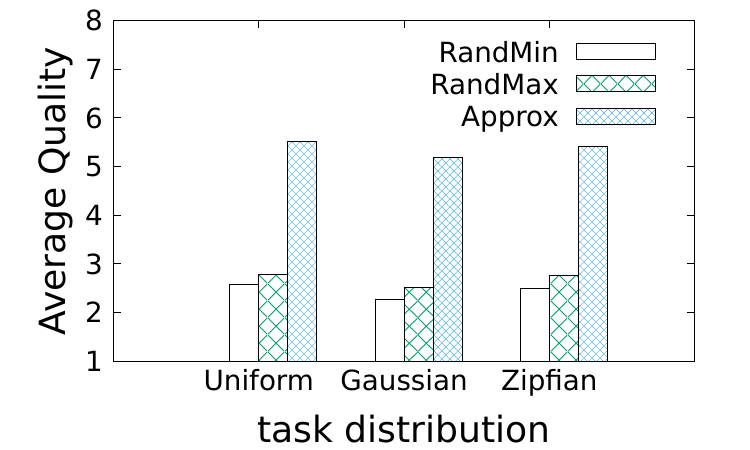}}
\subfigure[Quality vs. Budgets ($q_{min}$)] {\includegraphics[height=0.8in,width=1.7in]{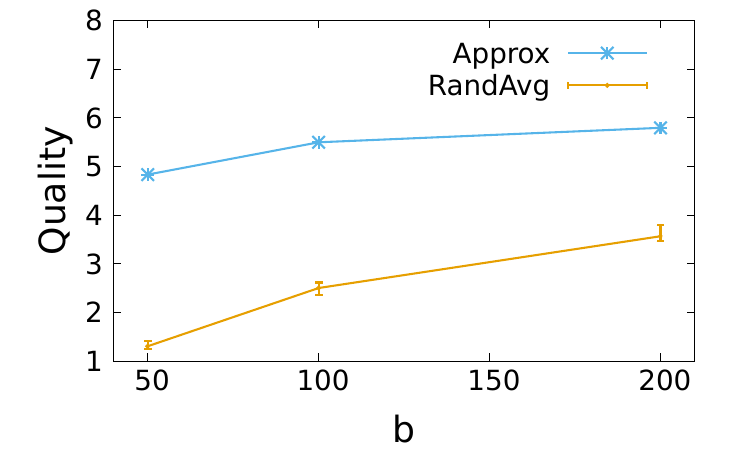}}
\caption{Quality of Multi-task Case}
\label{fig:multi-quality}
\vspace{-5pt}
\end{figure*}

\subsection{Results on Efficiency}
\label{ret:eff}

\begin{figure*}[ht]
\center
\vspace{-10pt}
\subfigure[Time of varying $m$] {\includegraphics[height=0.8in,width=1.7in]{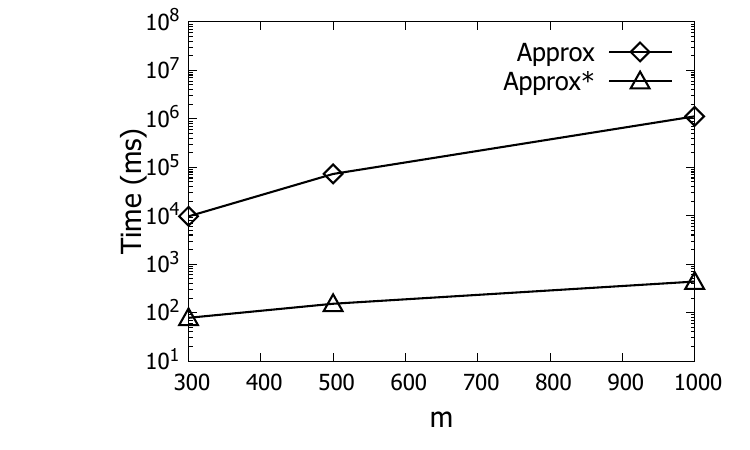}}
\subfigure[Time of varying $|W|$] {\includegraphics[height=0.8in,width=1.7in]{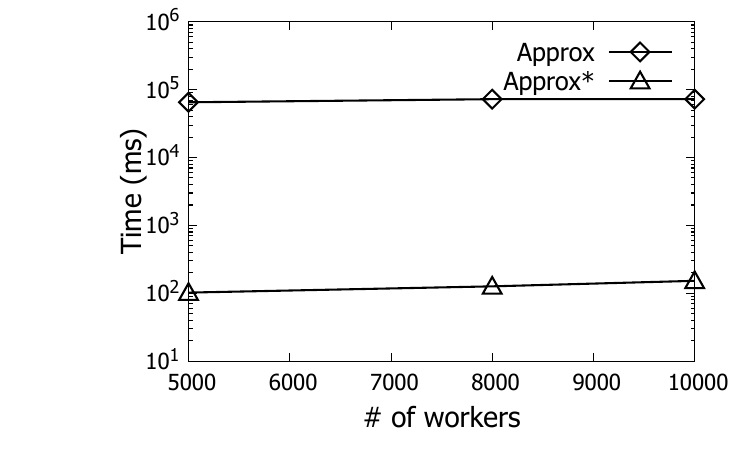}}
\subfigure[Time Breakdown]
{\includegraphics[height=0.8in,width=1.7in]{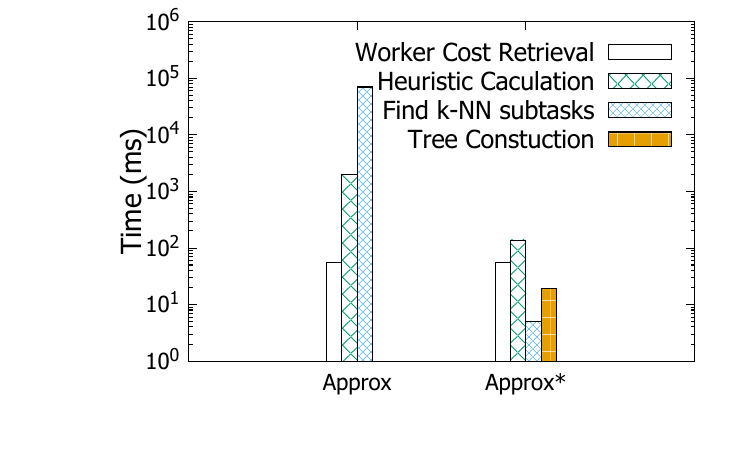}}
\subfigure[Analysis on Pruning Ratios] {\includegraphics[height=0.8in,width=1.7in]{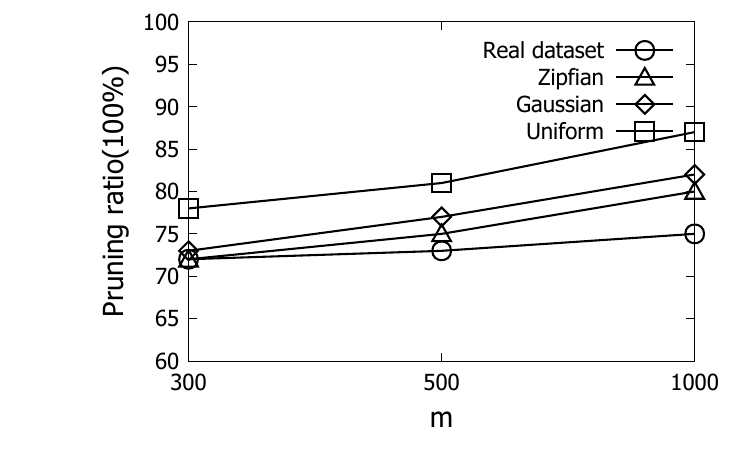}}
\vspace{-10pt}

\subfigure[Analysis on Tree Construction]
{\includegraphics[height=0.8in,width=1.7in]{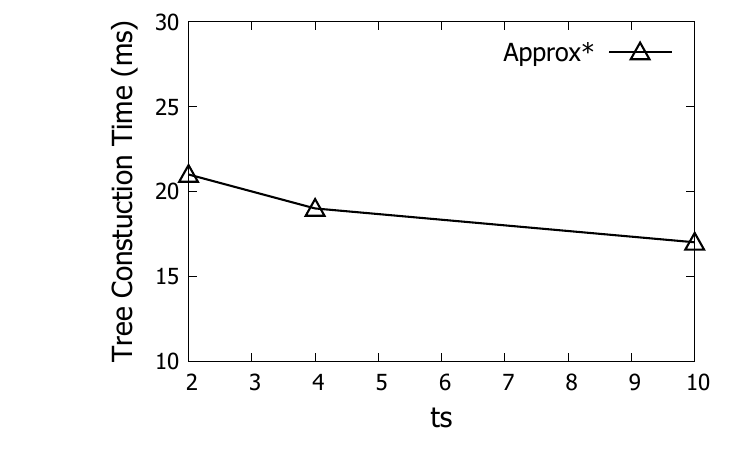}}
\subfigure[Effect of Distributions] {\includegraphics[height=0.8in,width=1.7in]{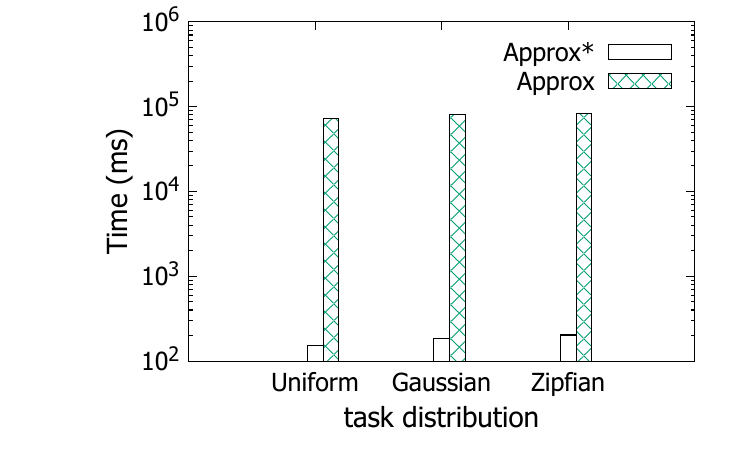}}
\subfigure[Effect of $k$]
{\includegraphics[height=0.8in,width=1.7in]{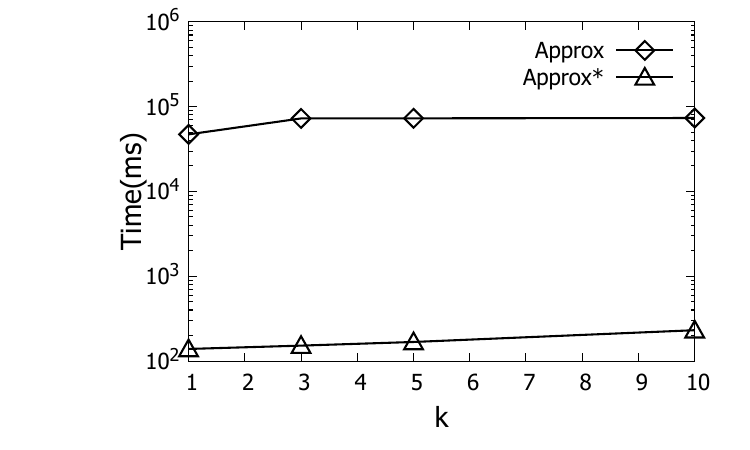}}
\subfigure[Effect of Budgets] {\includegraphics[height=0.8in,width=1.7in]{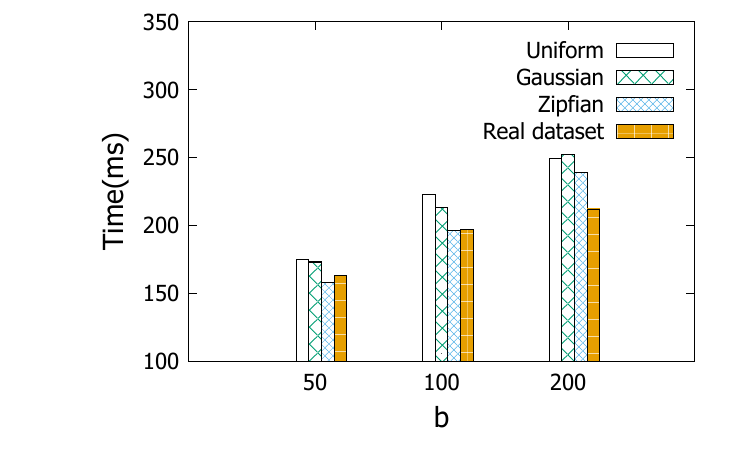}}
\vspace{-5pt}
\caption{Results of single-task assignment}
\label{fig:single-assignment}
\vspace{-5pt}
\end{figure*}



We examine the efficiency and scalability of our proposal, by comparing two variants Approx and Approx*.
The Approx solution is described by Algorithm~\ref{alg:single}, but without optimization techniques in Section~\ref{subsec:eff}. Approx* improves Approx by: 1) using tree-structured order-$k$ Voronoi diagrams to avoid redundant $k$-NN pre-computation; 2) using best-first searching and upper bound pruning for identifying the one with largest heuristic value.
We test the efficiency in single-task assignment in Fig.~\ref{fig:single-assignment} and multi-task assignment in Fig.~\ref{fig:multi-assignment}.

{\bfseries Single task assignment.}
First, we test the efficiency of Approx and Approx*, by varying the number of subtasks ($m$) in Fig.~\ref{fig:single-assignment} (a).
Approx* improves over Approx by two orders of magnitude. As $m$ increases, the improvement is more significant.
It shows that the optimization techniques, i.e.,
the tree-structured order-$k$ Voronoi diagram (Section~\ref{subsec:eff}), bring in good scalability to the approximation framework.
Second, we test the efficiency by varying the number of workers in Fig.~\ref{fig:single-assignment} (b).
The time cost keeps stable and increases only slightly w.r.t. $|W|$.
The reasons are two-fold: 1) the increasing trend is moderate due to the good scalability of best-first NN searching algorithm; 2) the slight increase shows that, with larger $|W|$, the completion ratios of tasks increase, and are with higher costs. In all cases, Approx* outperforms Approx by at least two orders of magnitude, showing good efficiency and scalability in terms of $m$ and $|W|$.

To understand how the efficiency is achieved, we make detailed analysis in Fig.~\ref{fig:single-assignment} (c-e).
The improvements made by Approx* are in two parts, as shown in Fig.~\ref{fig:single-assignment} (c). First, Approx* utilizes the implementation of the approximation of order-$k$ Voronoi diagram, and thus maximally reuses the computation of $k$-NN results. It can be observed that the cost of the interpolation (i.e., finding $k$ nearest subtasks) can be reduced by $4$ orders of magnitude.
Second, the tree-based pruning techniques can further reduce the cost of heuristic value calculation by more than an order of magnitude. 
The little extra cost for Approx* on the tree-structure is well-spent, given
the efficiency gained in the total execution time.

To examine the pruning effects supported by the tree structure, we report the pruning ratios, by varying $m$ on different task distributions in Fig.~\ref{fig:single-assignment} (d).
The ratio is calculated by the dividing the number of slots executed with pruning (Section~\ref{subsec:eff}) by the one without pruning.
It can be observed that our methods prune away more than 70\% subtask execution and therefore effectively accelerate the entire task processing. Similar trend is observed for the result on the real data.
We report the time cost spent on the tree-structure construction by varying the value of fanout of the tree structure, $t_s$, in Fig.~\ref{fig:single-assignment} (e).  In all testing, the construction time is no more than $25$ ms.
Also, the time decreases w.r.t. the increase of $t_s$, since a larger $t_s$ corresponds to a smaller number of tree nodes and therefore less construction time.

We continue to examine the effects of other factors on the efficiency.
In Fig.~\ref{fig:single-assignment} (f), we compare the two solutions by varying the distributions of tasks' locations.
In all cases, the performance of Approx* dominates that of Approx by more than two orders of magnitude.
More, the time cost of Approx* remains relatively stable with tasks' location distributions.
We also test the effect of parameter $k$ for data interpolation, in Fig.~\ref{fig:single-assignment} (g).
The time cost increases with $k$, since the cost of $k$-NN interpolation is higher for a bigger $k$.
We study the effect of budgets in Fig.~\ref{fig:single-assignment} (h).
 The time cost increases moderately w.r.t. $b$, since the number of executed subtasks also increases w.r.t. $b$.
Zipfian distribution has the lowest construction time.
A task tends to incur higher cost under skewed distributions, so that
the number of executed subtasks is reduced and the corresponding time cost is less. In summary, Approx* dominates Approx in different parameter settings, and has better adaptivity to the skewness of data distributions.

\begin{figure*}[ht]
\center
\subfigure[Time vs. number of cores] {\includegraphics[height=0.8in,width=1.7in]{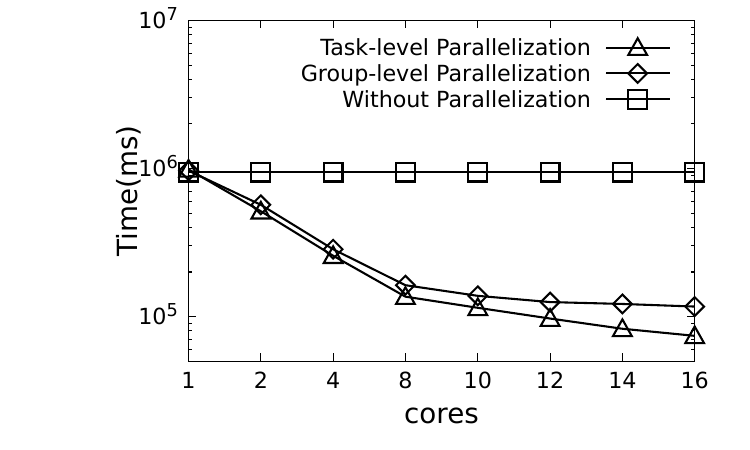}}
\subfigure[Time vs. distributions] {\includegraphics[height=0.8in,width=1.7in]{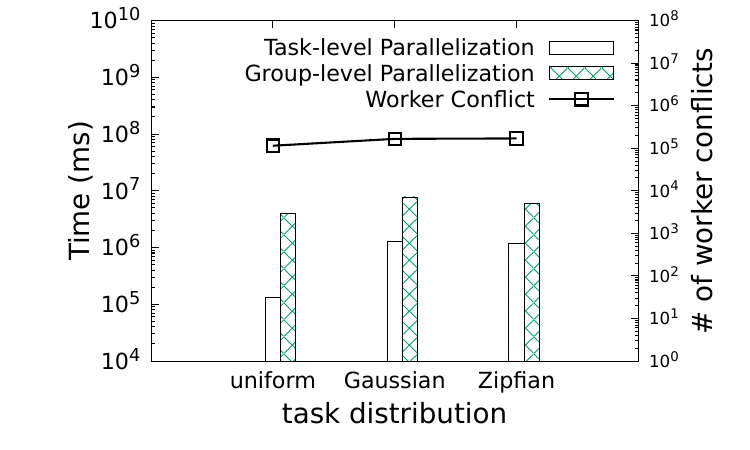}}
\subfigure[Conflicts vs. $|\mathcal{T}|$] {\includegraphics[height=0.8in,width=1.7in]{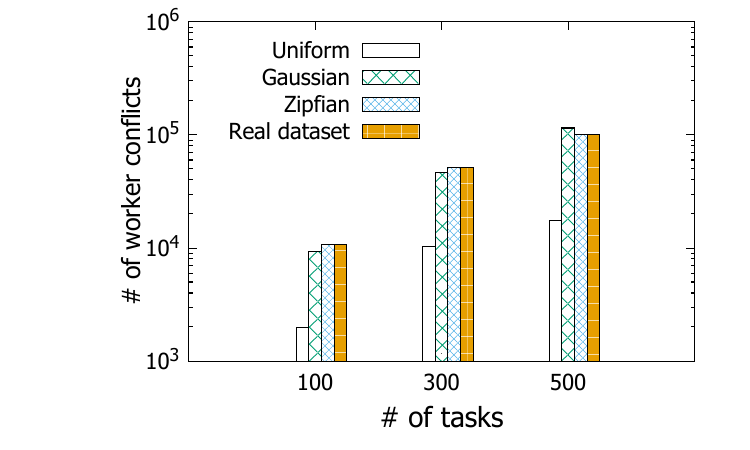}}
\subfigure[Time vs. $|\mathcal{T}|$] {\includegraphics[height=0.8in,width=1.7in]{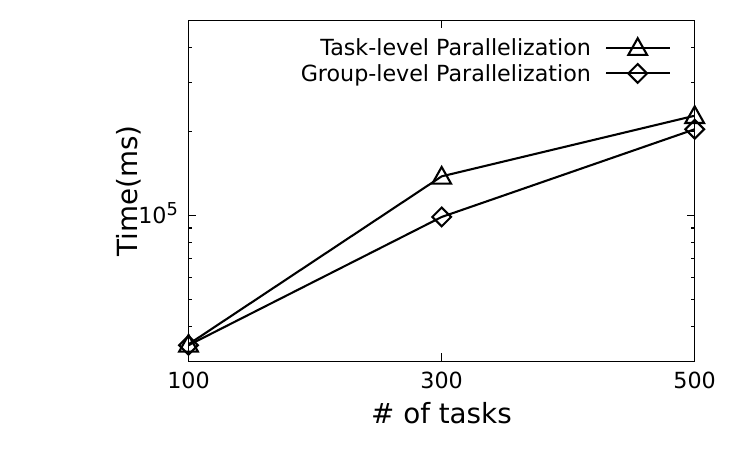}}
\subfigure[Time vs. $m$] {\includegraphics[height=0.8in,width=1.7in]{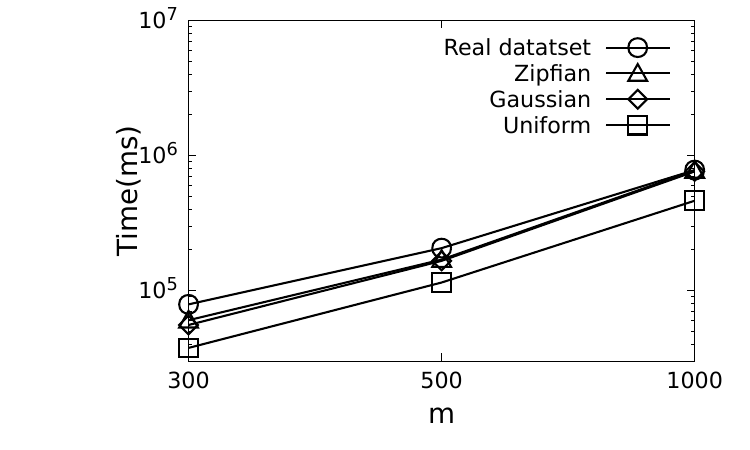}}
\subfigure[Time vs. number of cores (Priority)] {\includegraphics[height=0.8in,width=1.7in]{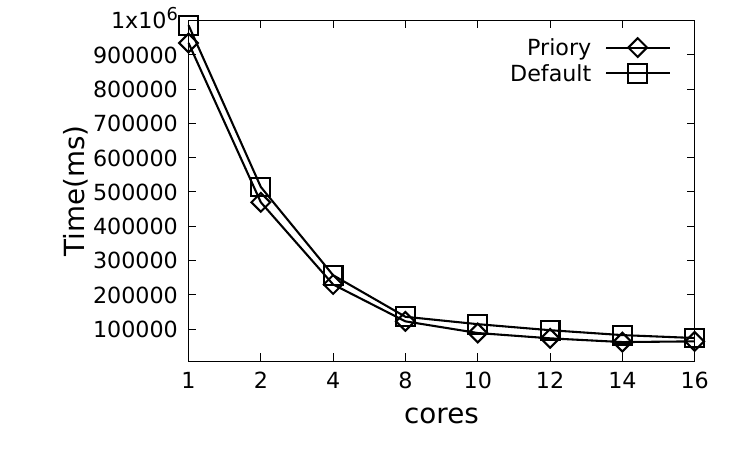}}
\subfigure[Time vs. $|\mathcal{T} |$ (MMQM)] {\includegraphics[height=0.8in,width=1.7in]{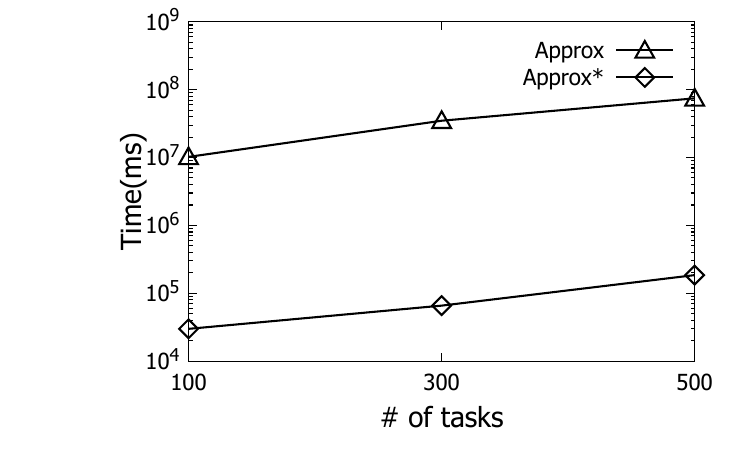}}
\subfigure[Time vs. $m$  (MMQM)] {\includegraphics[height=0.8in,width=1.7in]{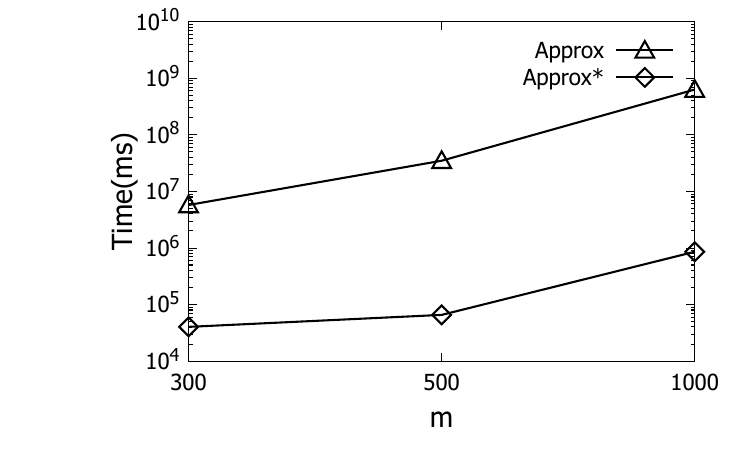}}
\caption{Results of multi-task assignment}
\label{fig:multi-assignment}
\end{figure*}

{\bfseries Multiple task assignment.}
We provide the results on the summation quality case in Fig.~\ref{fig:multi-assignment} (a-f)
 and the results on the minimum quality case in Fig.~\ref{fig:multi-assignment} (g-h).

First, we compare the performance of the three variants, group-based parallelization, task-based parallelization, and the basic solution without parallelization in Fig.~\ref{fig:multi-assignment} (a).
It shows that parallelization with smaller granularity achieves better scaleup.
The task-based parallelization outperforms the other two.
In particular, when the number of cores reaches $10$, the task-based parallelization solution takes about an order of magnitude less running time than the basic solution, which is consistent with the analysis in Section~\ref{subsec:maxsum}.
We then compare the two parallelization methods by varying task location distributions in Fig.~\ref{fig:multi-assignment} (b).
We can see that the Gaussian and Zipf distributions incur higher costs.
It is because that skewed datasets tend to incur larger numbers of worker conflicts. The cost increases moderately because of the good scalability achieved by the optimization techniques for indexing and scheduling.
Also, the number of worker conflicts increases with the number of tasks, as reported in Fig.~\ref{fig:multi-assignment} (c). 

We then report the scalability of our proposal w.r.t. the number of tasks in Fig.~\ref{fig:multi-assignment} (d). It can be observed that the task-based parallelization solution increases moderately w.r.t. the number of tasks, and grows only slowly. We examine the performance of the algorithm by varying the parameter $m$, in Fig.~\ref{fig:multi-assignment} (e).
All methods increase moderately, where Zipfian and Gaussian distributions take longer than the uniform case. The result is consistent with our analysis, because skewed task distributions have a larger chance for incurring worker conflicts.
We then make more analysis of the task-based parallelization method. We test the effect of thread priority setting in Fig.~\ref{fig:multi-assignment} (f).
It reflects that threads with lower heuristic values are scheduled to execute earlier by the priority adjustment module, which breaks the ties of blocked threads and improves the performance of parallelization.

We show the results on the minimum quality metric in Fig.~\ref{fig:multi-assignment} (g) and (h).
First, we examine the time costs by varying $|\mathcal{T}|$ in Fig.~\ref{fig:multi-assignment} (g).
It can be observed that the time cost increases w.r.t. the number of tasks. 
Second, we test the result by varying $m$ in Fig.~\ref{fig:multi-assignment} (h).
It can be observed that the running time increases as $m$ increases. 
In both experiments, Approx* steadily outperforms Approx, demonstrating better scalability in terms of the number of tasks and subtasks.




\section{Related Work}
\label{sec:related}

There are many studies in spatial crowdsourcing, requiring workers traveling to locations of spatial tasks and performs tasks, such as taking photos/videos, repairing a house, and waiting in line at shopping malls. These works focused on assigning available workers to tasks with distinct goals, such as maximizing the number of assigned tasks\cite{kazemi2012geocrowd} \cite{to2015server} \cite{DBLP:journals/pvldb/ChengJC18}, minimizing the total travel costs of all workers \cite{deng2015task} \cite{DBLP:journals/tetc/ZhangYLT19}, maximizing the quality score \cite{cheng2015reliable} \cite{to2015server} \cite{DBLP:journals/csur/DanielKCBA18}, maximizing quality task assignment by considering
both present and future workers/tasks \cite{DBLP:conf/icde/ChengLCS17}, or minimizing maximum task assignment delay \cite{DBLP:conf/icde/ChenCZC19}.
Existing works cannot be directly used to handle the quality issues of the TCSC problem.
According to our comprehensive survey \cite{todsurvey}, TCSC is related to the categories of data collection, and of task matching with quality constraints.
We thus review existing works in the two categories.

Regarding applications of data collection, there exist papers on floorplan generation~\cite{Alzantot2012a}, traffic anomalies detection~\cite{Pan2013}, voluntary services \cite{wuv}, and geo-spatial linked open data postprocessing~\cite{Karam2013b}, etc.
They mostly collect data in a specified spatiotemporal context, and do not address the issues in long-term data acquisition.
Our work can support extending these works for the continuous data acquisition, e.g., monitoring routing behaviours, by incorporating the quality-aware crowdsourcing framework.

Regarding quality constraints, most existing papers are on the quality of task responses, based on the workers' expertise, reputation, or reliability \cite{DBLP:journals/pvldb/LiuLOSWZ12}\cite{cheng2016task}\cite{cheng2015reliable} \cite{DBLP:conf/gis/KazemiSC13}.
They usually involve a pre-task qualification test \cite{DBLP:journals/pvldb/LiuLOSWZ12}, or the assignment based on the expertise \cite{cheng2016task}, or abilities \cite{cheng2015reliable} \cite{DBLP:conf/gis/KazemiSC13} of the worker. These papers are similar to the TCSC problem in the sense that they require data aggregation from multiple workers, but the aggregation methods are totally different.
To our best knowledge, the most relevant work is \cite{cheng2015reliable}, which considers the diversity (or distribution) of spatial and temporal tasks.
Differently, they do not consider the mutual interaction between the interpolated and crowdsourced data. So, the optimization target and corresponding techniques are totally different.

To summarize, a TCSC task has a temporally continuous nature, and requires time-sharing collaboration of multiple workers, necessitating quality-aware data management.
\section{Conclusion}
\label{sec:conclusion}
In this paper, we study the problem of TCSC, which enables time-sharing collaboration among multiple workers towards long-term continuous spatial crowdsourcing applications. We propose an entropy-based quality metric for measuring the incompleteness of the crowdsourced results. Based on that, we study quality-aware task assignment algorithms with budget constraints for both single- and multi-task cases.
For both variants, we prove its NP-hardness and submodularity of quality functions, so that a unified approximation framework can be applied. We devise novel indexing and parallel mechanisms for accelerating the processing.
Extensive experiments on real and synthetic datasets show that our proposals achieve good efficiency and scalability.
In the future, we will extend the approximation framework and optimization techniques from supporting temporal interpolation to spatiotemporal interpolation scenarios.


\section*{Acknowledgments}
We thank anonymous reviewers for their insightful comments on crowdsourcing with spatiotemporal interpolation.

\bibliographystyle{IEEEtran}
\bibliography{ms}

\balance

\newpage


\section{Appendix}


\subsection{Submodularity and Non-decreasingness of $p^{(j)}$}

\label{sec:sub_non}
{
We show Lemmas~\ref{lem:psub} and \ref{lem:pinc}, which are on the properties of finishing probability functions.
Assume a set $S$ of executed subtasks, and a to-be-executed subtask $e$, satisfying $S \cap \left\{e\right\} = \emptyset$.
We define $p^S(\tau^{(j)})$ as the finishing probability of subtask $\tau^{(j)}$ given that the subtasks in $S$ are executed.

Similarly, we define $\rho_{err}^S(\tau^{(j)})$ and $I^S(\tau^{(j)})$ as the error ratio and the interpolation distance (i.e., $I^S(\tau^{(j)}) =\sum_{e \in S_{KNN}} |\tau^{(j)}, e|_i$) of subtask $\tau^{(j)}$, respectively, given that all subtasks in $S$ are executed.
Without causing any ambiguities, in the proofs, we simplify $\rho_{err}^S(\tau^{(j)})$ and $I^S(\tau^{(j)})$ as $\rho_{err}^S$ and $I^S$, respectively.
\begin{lemma}
\label{lem:psub}
The function $p^{(j)}$ is submodular.
\end{lemma}
\begin{proof}
To prove function $p^{(j)}$ is submodular, it is equivalently to prove the following.
\begin{equation}
\label{proof:pj}
p^{ S \cap \left\{ e \right\} } (\tau^{(j)}) +  p^{S \cup \left\{ e \right\}} (\tau^{(j)}) \leq p^S (\tau^{(j)}) + p^{ \left\{ e \right\} } (\tau^{(j)})
\end{equation}
By substituting it with Equation~\ref{eqn:p}, we can rewrite it as $\rho_{err}^{ S \cap \left\{ e \right\} }  +  \rho_{err}^{S \cup \left\{ e \right\}}  \geq \rho_{err}^S  + \rho_{err}^{ \left\{ e \right\} }
$.
Equivalently, it is sufficient to prove that
\begin{equation}
\label{pf:rho}
\rho_{err}^{ S \cap \left\{ e \right\} }  +  \rho_{err}^{S \cup \left\{ e \right\}} - \rho_{err}^S  - \rho_{err}^{ \left\{ e \right\} } \geq 0
\end{equation}
From Equation~\ref{eqn:rho}, we know that $\rho_{err}$'s value is dependent on the $\tau^{(j)}$'s $k$-NN set, $S_{KNN}$, which must be a subset of $S$.
Next, we show the correctness of Equation~\ref{pf:rho}, by enumerating all three possible case of $S$.

{\it Case} 1.  When $S = \emptyset$, we have
$\rho_{err}^{ S \cap \left\{ e \right\} }  +  \rho_{err}^{S \cup \left\{ e \right\}} - \rho_{err}^S  - \rho_{err}^{ \left\{ e \right\} } = 0$, and thus Equation \ref{pf:rho} holds.

{\it Case} 2.  When $0<|S|<k$,  we have
$\rho_{err}^{ S \cap \left\{ e \right\} } = 1$ and $\rho_{err}^{ \left\{ e \right\} } = 1 - \frac{1}{k} + \frac{I^{ \left\{ e \right\} }}{km}$.
Here, $\rho_{err}^{S \cup \left\{ e \right\}}$ has two subcases, depending on $|S|$.

The subcase (a) refers to $|S|=k-1$. In that case, there are $k$ finished subtasks after executing $e$. Then, $ \rho_{err}^{S \cup \left\{ e \right\}} = \frac{I^{S \cup \left\{ e \right\}}}{km}$, and $\rho_{err}^S = \frac{I^S}{km} +  \frac{1}{k}$. We can have $\rho_{err}^{ S \cap \left\{ e \right\} }  +  \rho_{err}^{S \cup \left\{ e \right\}} - \rho_{err}^S  - \rho_{err}^{ \left\{ e \right\} } =  \frac{I^{S \cup \left\{ e \right\}}}{km} - \frac{I^S}{km} - \frac{I^{\left\{ e \right\}}}{km} = 0$.
Equation~\ref{pf:rho} holds.

The subcase (b) means that the total number of finished subtasks does not exceed $k$, after executing subtask $e$. Thus, we have $ \rho_{err}^{S \cup \left\{ e \right\}} = 1 - \frac{|S|+1}{k} + \frac{I^{S \cup \left\{ e \right\}}}{km}$ and $\rho_{err}^S = 1 - \frac{|S|}{k} + \frac{I^S}{km}$.
Then, $\rho_{err}^{ S \cap \left\{ e \right\} }  +  \rho_{err}^{S \cup \left\{ e \right\}} - \rho_{err}^S  - \rho_{err}^{ \left\{ e \right\} } = \frac{I^{S \cup \left\{ e \right\}}}{km}  - \frac{I^S }{km} - \frac{I^{\left\{ e \right\}}}{km}  = 0$. Equation \ref{pf:rho} holds.

{\it Case} 3. When $|S| \geq k$,  $\rho_{err}^{ S \cap \left\{ e \right\} } = 1$, we have $\rho_{err}^{ \left\{ e \right\} } = 1 - \frac{1}{k} + \frac{I^{ \left\{ e \right\} }}{km}$, $\rho_{err}^S = \frac{I^S}{km}$, and $ \rho_{err}^{S \cup \left\{ e \right\}} = \frac{I^{S \cup \left\{ e \right\}}}{km} $.
Based on whether the execution of subtask $e$ changes $S_{KNN}$ ($\tau^{(j)}$'s $k$-NN set), there can be two subcases.

If $S_{KNN}$ is not affected by $e$, we have that $ \rho_{err}^{S \cup \left\{ e \right\}} =\rho_{err}^S$. So, $\rho_{err}^{ S \cap \left\{ e \right\} }  +  \rho_{err}^{S \cup \left\{ e \right\}} - \rho_{err}^S  - \rho_{err}^{ \left\{ e \right\} } = \rho_{err}^{ S \cap \left\{ e \right\} } - \rho_{err}^{ \left\{ e \right\} } =  \frac{1}{k} - \frac{I^{ \left\{ e \right\} }}{km}$.
As the interpolation distance $I^{ \left\{ e \right\} }$ is less than $m $, we have $\frac{1}{k} - \frac{I^{ \left\{ e \right\} }}{km} >0$.
So, Equation \ref{pf:rho} holds.

If $S_{KNN}$ is affected by $e$, it implies that a subtask in $S_{KNN}$ is updated by $e$. Suppose the replaced subtask in original $S_{KNN}$ be $e'$, and the updated interpolation distance be $ I^{S \cup \left\{ e \right\}} = I^S - I^{\left\{ e' \right\}} + I^{ \left\{ e \right\} }$.
We can have $\rho_{err}^{ S \cap \left\{ e \right\} }  +  \rho_{err}^{S \cup \left\{ e \right\}} - \rho_{err}^S  - \rho_{err}^{ \left\{ e \right\} } = \frac{I^{S \cup \left\{ e \right\}}}{km} - \frac{I^S}{km} + \frac{1}{k} - \frac{I^{ \left\{ e \right\} }}{km} =  \frac{1}{k} - \frac{I^{ \left\{ e' \right\} }}{km} $. As the interpolation distance $I^{ \left\{ e' \right\} }$ is less than $m$, we have $\frac{1}{k} - \frac{I^{ \left\{ e' \right\} }}{km} >0$. Equation \ref{pf:rho} holds.

In summary, Equation \ref{pf:rho} holds in all three cases. The lemma is proved.
\end{proof}
\begin{lemma}
\label{lem:pinc}
The function $p^{(j)}$ is non-decreasing.
\end{lemma}
\begin{proof}
We prove the finishing probability function $p^{(j)}$ is non-decreasing by showing that the error rate function $\rho_{err}(\tau^{(j)})$ is non-increasing. Or, equivalently, 
\begin{equation}
\label{pf:non_increase}
\rho_{err}^{S \cup \left\{ e \right\}} - \rho_{err}^S \leq 0
\end{equation}
There can be two possible cases for set $S$, $0 \leq |S|<k$ and $|S| \geq k$. We hereby prove the correctness of Equation~\ref{pf:non_increase} by considering the two cases.

{\it Case} 1. When $0 \leq |S|<k$,  $ \rho_{err}^{S \cup \left\{ e \right\}}$ have two subcases, depending on the size of set $S$.

The first subcase is for $|S|=k-1$, and thus the total subtasks number is $k$ after adding the executed subtask $e$. Then, $ \rho_{err}^{S \cup \left\{ e \right\}} = \frac{I^{S \cup \left\{ e \right\}}}{km}$, and $\rho_{err}^S = \frac{I^S}{km} +  \frac{1}{k}$. So, we have $\rho_{err}^{S \cup \left\{ e \right\}} - \rho_{err}^S = \frac{ I^{S \cup \left\{ e \right\}}}{km} - \frac{I^S}{km} - \frac{1}{k} = \frac{I^{\left\{ e \right\}}}{km} - \frac{1}{k}$. As the interpolation distance $I^{ \left\{ e \right\} }$ is less than $m $,  Equation \ref{pf:non_increase} holds.

The second subcase is for $|S| < k-1$, meaning that the total number of subtasks is less than $k$ after the execution of $e$.
We can thus have $ \rho_{err}^{S \cup \left\{ e \right\}} = 1 - \frac{|S|+1}{k} + \frac{I^{S \cup \left\{ e \right\}}}{km}$, and  $\rho_{err}^S = 1 - \frac{|S|}{k} + \frac{I^S}{km}$. So, $ \rho_{err}^{S \cup \left\{ e \right\}} - \rho_{err}^S  = \frac{I^{S \cup \left\{ e \right\}}}{km}  - \frac{I^S}{km}  - \frac{1}{k} = \frac{I^{\left\{ e \right\}}}{km} - \frac{1}{k}$. Equation \ref{pf:non_increase} holds.

{\it Case} 2. When $S \geq k$, we can have $\rho_{err}^S = \frac{I^S}{km}$ and $ \rho_{err}^{S \cup \left\{ e \right\}} = \frac{I^{S \cup \left\{ e \right\}}}{km} $.
Based on whether the execution of subtask $e$ changes $S_{KNN}$, there can be two subcases.

If $S_{KNN}$ is not affected by $e$, we have $ \rho_{err}^{S \cup \left\{ e \right\}} =\rho_{err}^S$, so that Equation \ref{pf:non_increase} holds.

If $S_{KNN}$ is affected by $e$, it means that a subtask in $S_{KNN}$ is updated by $e$. We denote the replaced subtask in original $S_{KNN}$ as $e'$, and the updated interpolation distance as $ I^{S \cup \left\{ e \right\}} = I^S - I^{\left\{ e' \right\}} + I^{ \left\{ e \right\} }$.
Then, we can get $\rho_{err}^{S \cup \left\{ e \right\}} - \rho_{err}^S  = \frac{I^{S \cup \left\{ e \right\}}}{km} - \frac{I^S}{km}  =  \frac{I^{ \left\{ e \right\} }}{km} - \frac{I^{ \left\{ e' \right\} }}{km} $.
More, the fact that $e'$ is replaced by $e$ implies that $I^{ \left\{ e \right\} } < I^{ \left\{ e' \right\} }$. So, Equation \ref{pf:non_increase} holds.

In summary, Equation \ref{pf:non_increase} holds in all possible cases. Hence, the lemma is proved.
\end{proof}
}
\subsection{Proof of Lemma~\ref{lem:cond1}}
\label{subsec:cond1}

{
\begin{lemma}
\label{lem:cond1}
For a time segment $[l, r]$, if $knn(l)=knn(r)$, it is true that $\forall e \in [l, r], knn(e) = knn(l) = knn(r)$.
\end{lemma}
\begin{proof}
We prove that by contradiction. Assume a time segment $[l,r]$, $NN(l) = NN(r)$, and a slot $e$ on the segment, $e \in [l,r]$, $NN(l) \neq NN(e)$. We use $a$ and $b$ denote $NN(l)$ and $NN(e)$, $a \neq b$, $|b,e|_t <|a, e|_t$, and $b$ is in the range of $(e-|a,e|_t, e+|a,e|_t)$(It doesn't include two endpoints). According to the position of $a$, it can be divided into three cases:

The first case is for $l \leq a \leq r$. As $a$ is the $NN$ of $l$ and $r$, there is no other executed subtask within the time segment $[l-|l,a|_t, r+|r,a|_t]$. It has $e-|a,e|_t - (l-|l,a|_t) = |l,e|_t  -|a,e|_t + |l,a|_t \geq 0 $ and $r+|r,a|_t - (e+|a,e|_t) = |r,e|_t - |a,e|_t + |r, a|_t \geq 0 $, then the range of $b$ is within the range $[l-|l,a|_t, r+|r,a|_t]$. The case can not exist.

The second case is for $a < l$, and the range of $b$ is $(e, e+|a,e|_t)$. As $a$ is the $NN$ of $l$ and $r$, there is no other executed subtask within the time segment $(a, r + |a,r|_t)$, then the range of $b$ is within the range $(a, r + |a,r|_t)$. The case can not exist.

The third case is for $a > r$, then the range of $b$ is $(e-|a,e|_t, a)$. As $a$ is the $1-nn$ of $l$ and $r$, there is no other executed subtask within the time segment $(l - |a,l|_t, a)$, then the range of $b$ is within the range $(l - |a,l|_t, a)$. The case can not exist.
\end{proof}
}

\subsection{Extension to Spatiotemporal Interpolation}

{
{\it Spatiotemporal Interpolation.} Suppose a set of tasks $\mathcal{T}=\{\tau_1, \tau_2, ...\}$.
Each task $\tau_i$ consists of a set of $m$ subtasks, $\tau_i = \{\tau_i^{(j)}\}_{1 \leq j \leq m}$. If a subtask $\tau_i^{(j)}$ is not probed, it can either be {\it temporally interpolated} by the executed subtasks belonging to the same task $\tau_i$, or be spatially interpolated by subtasks satisfying that: 1) being executed at the same time slot $j$; 2) belonging to other tasks than $\{\tau_i\}$.
For example, in Fig.~\ref{fig:interpolation}, subtask $\tau_i^{(5)}$ can either be temporally interpolated by $\tau_i^{(3)}$ or spatially interpolated by $\tau_j^{(5)}$.

\begin{figure}
\centering
\includegraphics[width = 0.5\columnwidth]{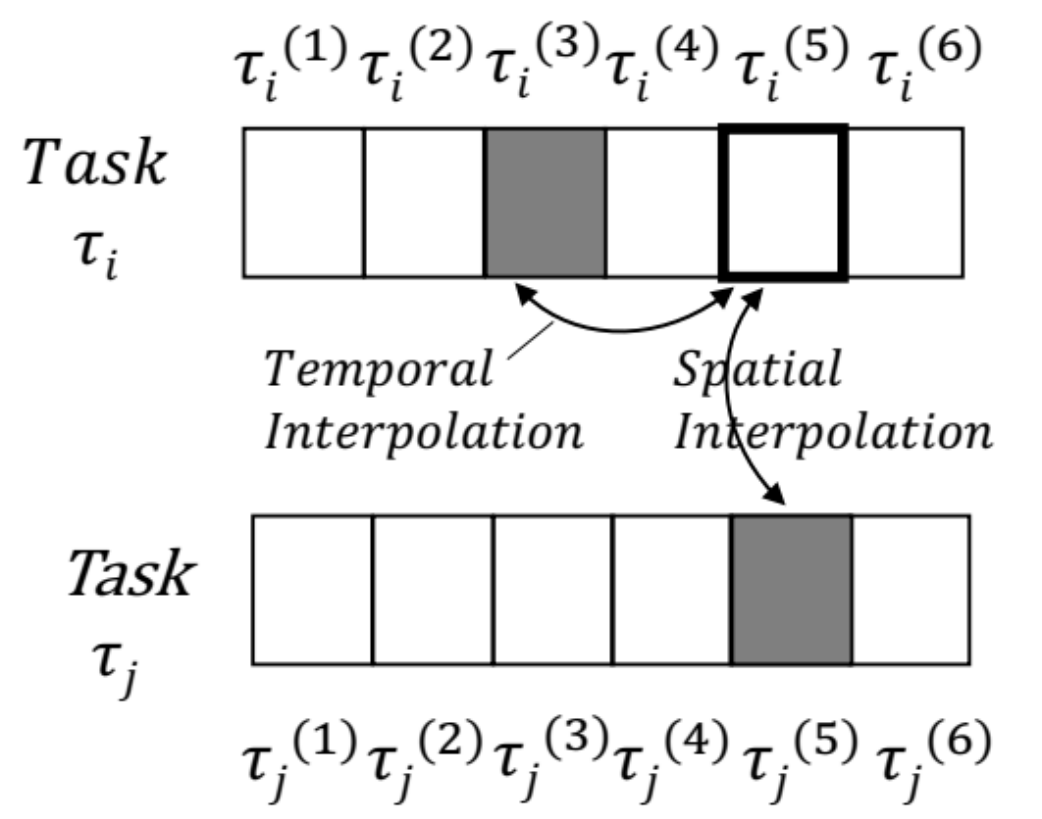}
\caption{Spatiotemporal Interpolation (Executed subtasks are shaded.)}
\label{fig:interpolation}
\vspace{-10pt}
\end{figure}

\begin{figure*}
\center
\vspace{-10pt}
\subfigure[Quality vs. Distributions] {\includegraphics[width=1.8in]{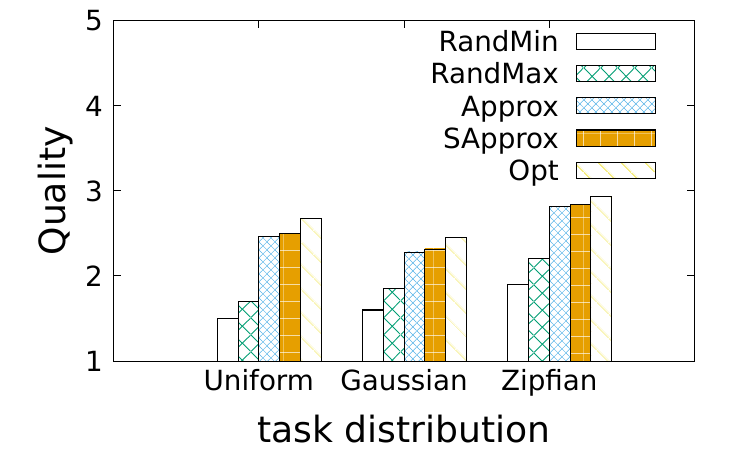}}
\subfigure[Quality vs. Budgets] {\includegraphics[width=1.8in]{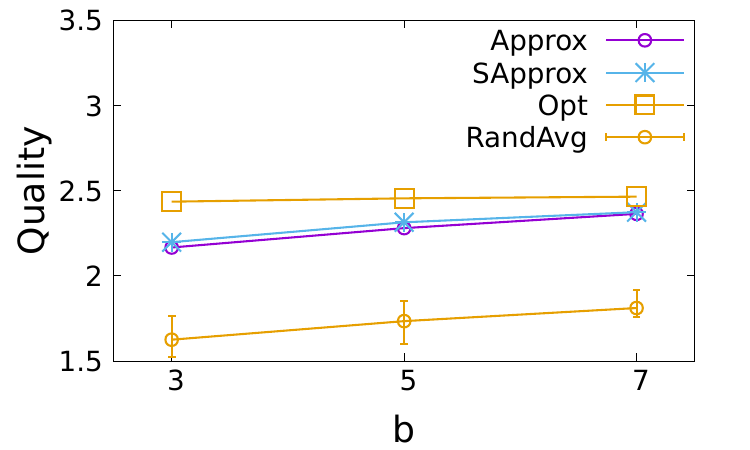}}
\subfigure[Quality w.r.t. Gaussian Distribution]
{\includegraphics[width=1.8in]{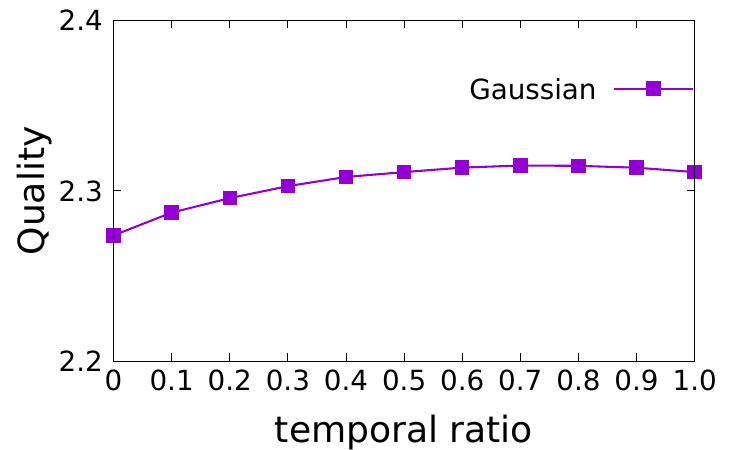}}
\vspace{-5pt}
\caption{Results with Spatiotemporal Interpolation}
\vspace{-10pt}
\label{ret:weighted}
\end{figure*}

{\it Extensions on Quality Metrics.} The spatial interpolation error is proportional to the spatial distances between the interpolated values and their neighboring values. The error ratio function $\rho^s_{err}(\tau_{i}^{(j)})$ for spatial interpolation can thus be written as follows.
\begin{align}
\vspace{-5pt}
\rho^s_{err}(\tau_i^{(j)}) = \frac{\sum_{e \in S^s_{kNN}(\tau_i^{(j)})} |\tau_i^{(j)},e|_e}{k\cdot |\mathcal{D}|}
\vspace{-5pt}
\end{align}
Here, $e$ represent an executed subtask, and $|\tau_i^{(j)},e|_e$ represents the spatial distance
between subtask $\tau_i^{(j)}$ and $e$.
Function $S^s_{kNN}(.)$ returns the $k$ executed subtasks with the smallest spatial distances. $|\mathcal{D}|$ in the denominator represents the spatial domain size so that the value range of the spatial interpolation error ratio $\rho^s_{err}$ is from $0$ to $1$, to be consistent with the form of temporal interpolation error ratio function.

We can use a weighted summation function to combine the interpolation errors of both spatial and temporal domains.
\begin{align}
\label{eqn:ratio}
\rho_{err} = w_s \cdot \rho_{err}^s + w_t \cdot \rho_{err}^t
\end{align}
Here, $w_s$ and $w_t$ are weights of the two components, whose sum equals $1$. $\rho_{err}^t$ represents the temporal interpolation error (Equation 3 in the manuscript).
So, the subtask finishing probability $p_i^{(j)}$ can be written as:
\begin{align}
p_i^{(j)}= \frac{1}{m}(1-\rho_{err}(\tau_i^{(j)})) \nonumber\\
= \underbrace{\frac{1}{m}(1- w_t \cdot \rho^t_{err}(\tau_i^{t(j)}))}_{\text{The temporal interpolation part}} +
\underbrace{\frac{1}{m}(1- w_s \cdot \rho^s_{err}(\tau_i^{s(j)}))}_{\text{The spatial interpolation part}} - \frac{1}{m} \nonumber
\end{align}

Both temporal and spatial interpolation parts can be proved to be submodular and non-decreasing, following the proofs of Lemmas 6 and 7. So, the summation of the two parts preserves the submodularity and non-decreasingness, according to the properties of composite submodular functions (Lemma 1 in the manuscript).
Similarly, quality function
$q(\tau_i) = -\sum_{j=1}^{m}p_i^{(j)}\log_2\big{(} p_i^{(j)}\big{)}$
can be proved to be submodular and non-decreasing, since the entropy function is known as concave and non-decreasing.

{\it Extensions on Multiple Task Assignment.} Since the spatial interpolation process refers to the interactions between multiple TCSC tasks, we hereby examine the two variants of multi-task assignment scenarios, with aggregated quality metrics $q_{sum}$ and $q_{min}$ as the maximization targets, respectively.

\begin{problem} {\it Spatio-Temporal Continuous Crowdsourcing ({\bf STCC in short)}}
Given a set of tasks $\mathcal{T}=\{\tau_1, \tau_2, ...\}$, the problem is to find a task assignment for $\tau_i \in \mathcal{T}$, such that the summation quality $q_{sum}(\mathcal{T})= \sum_{i=1}^{|\mathcal{T}|}q(\tau_i \big{|} \tau_i \in \mathcal{T})$, or the minimum quality
$q_{min}(\mathcal{T})= min \left\{q(\tau_i) \big{|} \tau_i \in \mathcal{T} \right\}$ can be maximized with given budgets.
\end{problem}

In the settings of spatiotemporal interpolation, we can prove that the summation quality and minimum quality functions are still submodular and non-decreasing with Lemma 1, because 1) {\tt SUM} and {\tt MIN} are concave functions; 2) the extended quality function is proved to be submodular and non-decreasing. Therefore, the approximation framework of Algorithm 1 can be applied for handling the multi-task assignment scenario. The heuristic value is set as the increase of the quality metrics divided by the corresponding cost (of a tentatively selected subtask), following the same greedy strategy and approximation ratio.
To this end, the framework of approximation algorithm can be preserved.

{\it Experimental Results.}
We conduct experiments with the updated quality metric in Fig.~\ref{ret:weighted}, following the default setting of the manuscript. SApprox refers to the results with spatiotemporal interpolation and Approx refers to the results with only temporal interpolation. By default, we set $w_s$ and $w_t$ to $0.3$ and $0.7$, respectively, for SApprox. For Approx, the $w_s$ is set to $1$, since it does not do spatial interpolation.

Fig.~\ref{ret:weighted} (a) reports the quality values w.r.t. data distributions. It can be observed that both SApprox and Approx are very close to the optimal result, OPT. SApprox is better than Approx, because of the quality improvement made by spatial interpolation. We also test how the quality varies w.r.t. the budgets. In all testing, SApprox is better than Approx, and both the two have significant improvement over the baselines.
To examine the effect of parameter tuning on $w_s$ and $w_t$, we plot Fig.~\ref{ret:weighted} (c), where X-axis is for the value of $w_t$. It shows that when $w_t$ equals 0.7 the highest quality value is achieved. Therefore, in our experiments, $w_t$ is set to $0.7$, by default.

{\it Extensions on Indexing.}
To support efficient evaluation, current indexing techniques need to be redesigned.
For TCSC, the index structure is based on a one-dimensional Voronoi diagram.
For STCC, the index structure is based on a multi-dimensional weighted order-k Voronoi diagram. We should study how to approximate such a diagram with indexing structures, including node splitting and stopping conditions, index-based maximum heuristic value calculation, etc.
}


\end{document}